\RequirePackage[l2tabu, orthodox]{nag}		
\documentclass[10pt]{article}

\usepackage[T1]{fontenc}
\usepackage{lmodern}
\usepackage{microtype}

\usepackage{natbib}       									
\usepackage{amsmath}                        
\numberwithin{equation}{section}
\usepackage{graphicx}                       
\usepackage[leftcaption]{sidecap}												
\usepackage{subfigure}
\usepackage{enumitem}                       
\usepackage{mdwlist}												
\usepackage[dvipsnames]{xcolor}
\usepackage{hyperref}
\usepackage{amssymb}                        
\usepackage[mathscr]{eucal}                 
\usepackage{dsfont}													
\usepackage{accents}
\usepackage{pstricks}
\usepackage{rotating}
\usepackage{lscape}
\usepackage{cancel}
\usepackage[paperwidth=8.5in,paperheight=11in,top=1.25in, bottom=1.25in, left=1.00in, right=1.00in]{geometry}
\usepackage{mathtools}                      
\mathtoolsset{showonlyrefs=true}            
\usepackage{fixltx2e,amsmath}               
\MakeRobust{\eqref}

\usepackage{mathdots}
\usepackage{microtype}											
\usepackage{amsthm}                         
\allowdisplaybreaks                         
\theoremstyle{plain}
\newtheorem{theorem}{Theorem}
\numberwithin{theorem}{section}
\newtheorem{lemma}[theorem]{Lemma}          
\newtheorem{proposition}[theorem]{Proposition}
\newtheorem{corollary}[theorem]{Corollary}
\theoremstyle{definition}
\newtheorem{definition}[theorem]{Definition}

\newtheorem{remark}[theorem]{Remark}
\newtheorem{assumption}[theorem]{Assumption}

\renewcommand{\(}{\left(}
\renewcommand{\)}{\right)}

\newcommand\Eb{\mathds{E}}
\newcommand\Pb{\mathds{P}}

\newcommand\Rb{\mathds{R}}

\newcommand\Nb{\mathds{N}}

\newcommand\Ac{\mathscr{A}}
\newcommand\Bc{\mathscr{B}}
\newcommand\Cc{\mathscr{C}}

\newcommand\Fc{\mathscr{F}}
\newcommand\Gc{\mathscr{G}}
\newcommand\Hc{\mathscr{H}}
\newcommand\Lc{\mathscr{L}}

\newcommand\Oc{\mathscr{O}}
\newcommand\Pc{\mathscr{P}}

\newcommand\Xc{\mathscr{X}}
\newcommand\Yc{\mathscr{Y}}

\newcommand\eps{\varepsilon}

\newcommand\Om{\Omega}
\newcommand\sig{\sigma}

\newcommand\gam{\gamma}
\newcommand\Gam{\Gamma}
\newcommand\lam{\lambda}
\newcommand\del{\delta}

\newcommand\xb{\bar{x}}
\newcommand\yb{\bar{y}}

\newcommand\ub{\bar{u}}
\newcommand\qb{\bar{q}}

\newcommand\IIb{\bar{\II}}

\newcommand{\psib}{\bar{\psi}}
\newcommand{\etab}{\bar{\eta}}
\newcommand{\xib}{\bar{\xi}}

\newcommand\Cv{\mathbf{C}}
\newcommand\mv{\mathbf{m}}

\newcommand\Pbh{\widehat{\Pb}}
\newcommand\Ebh{\widehat{\Eb}}

\newcommand\Bh{\widehat{B}}

\newcommand\Ach{\widehat{\Ac}}

\newcommand\Ebt{\widetilde{\Eb}}
\newcommand\Pbt{\widetilde{\Pb}}
\newcommand\Act{\widetilde{\Ac}}

\newcommand\Gct{\widetilde{\Gc}}

\newcommand\Bt{\widetilde{B}}

\renewcommand\d{\partial}

\newcommand\II{\mathtt{I}}

\newcommand\dd{\mathrm{d}}
\newcommand\ee{\mathrm{e}}
\newcommand\BS{\textrm{BS}}
\newcommand\Ind{\textrm{Nonlin}}

\newcommand{\norm}[1]{\left\|#1\right\|}

\begin{document}


\title{Indifference prices and implied volatilities}

\author{
Matthew Lorig
\thanks{Department of Applied Mathematics, University of Washington, Seattle, WA, USA.}
}

\date{This version: \today}

\maketitle

\begin{abstract}
We consider a general local-stochastic volatility model and an investor with exponential utility.  For a European-style contingent claim, whose payoff may depend on either a traded or non-traded asset, we derive an explicit approximation for both the buyer's and seller's indifference price.  For European calls on a traded asset, we translate indifference prices into an explicit approximation of the buyer's and seller's implied volatility surface.  For European claims on a non-traded asset, we establish rigorous error bounds for the indifference price approximation.   
Finally, we implement our indifference price and implied volatility 
approximations in two examples.
\end{abstract}

\noindent
\textbf{Key words}: Indifference pricing, implied volatility, PDE asymptotics, local-stochastic volatility, Heston

%
%

\section{Introduction}
\label{sec:intro}
When markets are incomplete, there exist infinitely many equivalent martingale measures, under which options could feasibly be priced.  As such, for a given derivative asset, there is a range of possible no-arbitrage values.  The concept of \emph{indifference pricing}, first introduced in \cite{hodges-neuberger-1989}, provides an economically justified manner for selecting unique no-arbitrage prices in incomplete market settings.  Additionally, indifference pricing leads naturally to a bid-ask spread, an empirically observed phenomenon whereby sellers of a derivative asset ask for a higher price than buyer's are willing to pay for it.
For an overview of indifference pricing methods and applications we refer the reader to \cite{carmona2009indifference}.
\par
Despite the desirable features mentioned above, the widespread use of indifference pricing methods has been hampered by the computational complexity of this problem.  Typically, to compute indifference prices, one must find an explicit expression for the value function of an investor who, in an incomplete market setting, seeks to maximize expected utility.  The Hamilton-Jacobi-Bellman partial differential equation (HJB PDE) associated with the value function is nonlinear and rarely has an explicit solution.  
Complicating the problem is the fact that the investor's terminal wealth depends both on his trading strategy and on the (random) payoff of the derivative asset to be priced.
Because it is rarely possible to solve the HJB PDE explicitly, closed-form formulas for indifference prices are typically not available.
Cases in which explicit indifference prices are available are often
limited to options written on non-traded assets.  We mention, in particular, 
\cite{carmona-ludkovski-2006}, \cite{grasselli-hurd-2007}, \cite{henderson-2002}, \cite{leung-ludkovski-2012} and \cite{musiela2004example}.  In every one of these papers, a linearization technique, introduced by \cite{zariphopoulou2001}, reduces the nonlinear HJB PDE to a linear, parabolic Cauchy problem, which can then be solved explicitly.  The linearization method does not work when the claim is written on a traded asset, such as a stock or market index, nor does it work when volatility has a local component, as it would in the well-known SABR model of \cite{sabr}.
\par
For options on traded assets, \cite{sircar2004bounds} find bounds for indifference prices.  Additionally, they obtain asymptotic approximations for indifference prices in a fast mean-reverting volatility setting.  However, the pricing approximation is taken only to first order and, as a result, does not capture the bid-ask spread property of the indifference pricing mechanism.  More recently, \cite{sircar-sturm-2012} use indifference pricing methods to derive a small-time characterization of implied volatility in a stochastic volatility (SV) setting.
And \cite{kumar2015effect} derives asymptotics for indifference prices and implied volatilities in a fast mean-reverting SV setting.
\par
The main contributions of the present manuscript are two-fold.
(1) In a general local-stochastic volatility (LSV) setting, we derive an explicit approximation for the indifference price of a European-style asset, whose payoff may depend on either a traded or non-traded asset.  For options on a non-traded asset, we are able to establish rigorous error bounds for our indifference pricing approximation. 
(2) For call options on a traded asset, we translate our indifference pricing approximation into an explicit approximation of implied volatility.  In doing so, we obtain both buyer's and seller's implied volatility surfaces.
\par
The main mathematical tool we shall employ to accomplish the above tasks is a Taylor series expansion method, which was first developed in \cite{pagliarani2011analytical} to solve linear pricing PDEs in a scalar diffusion setting.  This method was later extended in \cite{lorig-pagliarani-pascucci-2} and \cite{lorig-pagliarani-pascucci-4} to include more general polynomial expansions and to price options on multidimensional diffusions.  Note, however, that these three papers do not attempt to solve nonlinear PDEs as we do here.
\par
The rest of this paper proceeds as follows:
in Section \ref{sec:merton} we introduce a general class of LSV models, define the indifference price for a European-style claim, and show how the indifference price is related to the solution of a pair of coupled nonlinear PDEs.
In Section \ref{sec:asymptotics} we develop an explicit asymptotic approximation for European calls written on a traded asset.
We then translate the indifference price approximation for calls into an explicit approximation of implied volatility.
In Section \ref{sec:Y} we consider options on a non-traded asset.  We derive explicit approximations for the indifference price and implied Share ratio and establish rigorous error bounds for the former.
Lastly, in Section \ref{sec:examples} we implement our approximation methods in two examples.

%
%

\section{Model and problem formulation}
\label{sec:merton}
Let $(\Omega,\Fc,(\Fc_t)_{t \geq 0},\Pb)$ be a complete filtered probability space.  The probability measure $\Pb$ represents the physical (i.e., observable) probability measure and the filtration $(\Fc_t)_{t \geq 0}$ represents the history of the market.  For simplicity, throughout this paper, we assume a frictionless market, zero interest rates and no dividends.
\par
We consider a single risky asset $S$ whose dynamics under $\Pb$ are described by the following two-dimensional stochastic differential equation (SDE):
\begin{align}
\begin{aligned}
S_t
	&=	\exp \big( {X_t} \big) , \\
\dd X_t
	&=	\Big( \mu(X_t,Y_t) - \tfrac{1}{2}\sig^2(X_t,Y_t) \Big)  \dd t + \sig(X_t,Y_t)  \dd B_t^X , \\
\dd Y_t
	&=	c(X_t,Y_t) \dd t + \beta(X_t,Y_t) \Big( \rho  \dd B_t^X + \sqrt{1- \rho^2}  \dd B_t^Y \Big) , 
\end{aligned} \label{eq:physical}
\end{align}
where $B^X$ and $B^Y$ are independent Brownian motions under $\Pb$ and $\rho \in [-1,1]$.
We have in mind that $Y$ is the driver of volatility in a local-stochastic volatility (LSV) setting.  However, more generally, $Y$ could represent any \emph{non-traded} quantity that is observable.  We assume that the coefficients $(\mu,\sig,c,\beta)$ are such that SDE \eqref{eq:physical} admits a unique strong solution.
\par
Let $W$ denote the wealth process of an investor who invests $\pi_t$ units of currency in $S$ at time $t$ and invests $(W_t-\pi_t)$ units of currency in a bond. 
Assuming the investment strategy is self-financing, the wealth process $W$ satisfies
\begin{align}
\dd W_t
	&=	\frac{\pi_t}{S_t} \dd S_t
	=		\pi_t \mu(X_t,Y_t) \dd t + \pi_t \sig(X_t,Y_t) \dd B_t^X , \label{eq:dW}
\end{align}
where we have used the fact that the risk-free rate of interest is zero.
\par
Assume now that the investor has an initial wealth $W_t = w$ at time $t$ and also owns $\nu$ European-style contingent claims, each with payoff $\varphi(X_T,Y_T)$.  Here $\nu$ may be negative, indicating that the investor has sold $|\nu|$ claims.  The investor then trades the underlying stock and bond so as to maximize his expected terminal utility $\Eb \, U(W_T + \nu \varphi(X_T,Y_T))$, where $U$ is the investor's \emph{utility function}.  Note that the argument of $U$ includes both the wealth $W_T$ the investor obtains from trading and the payoff he receives from the European option.
\begin{definition}
We define the \emph{investor's value function} $V$ as
\begin{align}
V(t,x,y,w,\nu)
	&:=	\sup_{\pi \in \Pi} \Eb_{t,x,y,w}  U( W_T + \nu  \varphi(X_T,Y_T) ) , \label{eq:v.def}
\end{align}
where $\Pi$, the set of \emph{admissible strategies}, is given by
\begin{align}
\Pi 
	&:= \Big\{ \pi : \Eb_{0,x,y,w} \int_0^T \pi_t^2 \sig^2 (X_t,Y_t) \dd t < \infty \Big\} .	
\end{align}
\end{definition}
\begin{definition}
\label{def:indiff.price}
We define the \emph{indifference price per claim} for $\nu$ European options as the unique solution $u \equiv u(t,x,y,w,\nu)$ of the equation
\begin{align}
V(t,x,y,w,0)
	&=	V(t,x,y,w-\nu u,\nu) . \label{eq:indiff.price} 
\end{align}
When $\nu$ is positive, we will refer to $u$ as the \emph{buyer's indifference price}.  When $\nu$ is negative, we will refer to $u$ as the \emph{seller's} or \emph{writer's indifference price}.
\end{definition}
\noindent
The meaning of the indifference price per unit claim should be clear from the definition.  
Given two choices:
(i) purchase $\nu$ claims for a price $u$ per claim and dynamically invest the remaining wealth in the stock $S$ and a bond, and
(ii) dynamically invest all wealth in the stock $S$ and a bond,
the investor would be indifferent, because he could achieve the same expected utility in both scenarios.
\par
To find the indifference price $u$, which is our main objective, we must first find the value function $V$.  
The Hamilton-Jacobi-Bellman partial differential equation (HJB PDE) associated with control problem \eqref{eq:v.def} is
\begin{align}
(\d_t + \Ac) V + \max_{\pi \in \Rb} \Ac^\pi V
	&=	0 ,  &
V(T,x,y,w,\nu)
	&=	U(w + \nu \varphi(x,y)) , \label{eq:hjb.pde}
\end{align}
where $\(\Ac + \Ac^\pi\)$ is the generator of $(X,Y,W)$ assuming a Markov investment strategy $\pi_t = \pi(t,X_t,Y_t,W_t)$.  Specifically, the operators $\Ac$ and $\Ac^\pi$ are given by
\begin{align}
\Ac
	&=	\( \mu(x,y) -\tfrac{1}{2} \sig^2(x,y) \) \d_x + \tfrac{1}{2} \sig^2(x,y) \d_x^2
			+ c(x,y) \d_y + \tfrac{1}{2} \beta^2(x,y) \d_y^2 + \rho \sig(x,y) \beta(x,y) \d_x \d_y, \\
\Ac^\pi
	&=	\pi(t,x,y,w) \mu(x,y) \d_w + \frac{1}{2} \pi^2(t,x,y,w) \sig^2(x,y) \d_w^2 \\ &\quad
			+ \pi(t,x,y,w) \rho \sig(x,y) \beta(x,y) \d_y \d_w 
			+ \pi(t,x,y,w) \sig^2(x,y) \d_x \d_w . \label{eq:A.pi}
\end{align}
Here, and throughout this manuscript, we use the notation $\d_x^n := \tfrac{\d^n}{\d x^n}$ and likewise for other partial derivatives.
We assume that the HJB-PDE \eqref{eq:hjb.pde} admits a unique classical solution, which coincides with the value function \eqref{eq:v.def}.
In general, uniqueness only holds within a class of functions satisfying certain growth conditions, which inevitably depend on the utility function $U$ and the option payoff $\varphi$.
The candidate optimal strategy $\pi^*$ is given by maximizing $\Ac^\pi V$.  We have
\begin{align}
\pi^*
	&=	\operatorname*{argmax}_{\pi \in \Rb} \Ac^\pi V
	=	-  \( \frac{ \mu \, (\d_w V) + \rho \beta \sig \, (\d_y \d_w V) + \sig^2 \, (\d_x \d_w V)}{\sig^2 \, \d_w^2 V} \)  . \label{eq:pi}
\end{align}
where, for simplicity (and from now on), we have omitted the arguments $(t,x,y,w)$.
Inserting the optimal strategy $\pi^*$ into the HJB-PDE \eqref{eq:hjb.pde} yields
\begin{align}
0
	&=	\( \d_t + \Ac \) V + \Hc (V) , &
V(T,x,y,w,\nu)
	&=	U\( w + \nu  \varphi(x,y) \) , \label{eq:hjb.3} 
\end{align}
where the \emph{Hamiltonian} $\Hc(V)$ is a nonlinear term, which is given by
\begin{align}
\Hc (V) 	
	&=	- \tfrac{1}{2}\lam^2 \frac{(\d_w V)^2}{\d_w^2 V}
			- \rho \beta \lam \frac{(\d_w V)(\d_y \d_w V)}{\d_w^2 V}
			- \tfrac{1}{2} \rho^2 \beta^2 \frac{(\d_y \d_w V)^2}{\d_w^2 V} 
			\\ &\quad
			- \mu \frac{(\d_w V)(\d_x \d_w V)}{\d_w^2 V}
			- \rho \sig \beta \frac{(\d_x \d_w V)(\d_y \d_w V)}{\d_w^2 V}
			- \tfrac{1}{2} \sig^2 \frac{(\d_x \d_w V)^2}{\d_w^2 V} , & 
\lam
	&=	\frac{\mu}{\sig} . 
\end{align}
Note that we have introduced $\lam$, the \emph{Sharpe ratio}.
\begin{remark}[On notation]
While it is standard in stochastic control literature to use subscripts to indicate partial derivatives (e.g., $V_x := \d_x V$), it is standard to perturbation theory to use subscripts to indicate the order of a given term in powers of some small parameter.  For example, if $\eps$ is a the small parameter of interest, then terms that are of order $\Oc(\eps^n)$ carry a subscript $n$.  In this paper, we follow the standard in perturbation literature and use subscripts to keep track of the order of terms in powers of a small parameter $\eps$, which we shall introduce in the next section.  The advantage of this approach is that each equation in the asymptotic analysis that follows can be easily checked for consistency by summing the subscripts of a given term and checking that the total is equal to the order of a given equation.  For example, an equation of order $\Oc(\eps^3)$ may contain terms such as $\d_x V_3$ and $(\d_x V_2)(\d_y V_1)$, but it will not contain terms such as $V_4$ or $(\d_y V_2) (\d_x V_2)$.
\end{remark}
\noindent
To proceed further, we must assume a specific form for the investor's utility function $U$.
\begin{assumption}
We shall assume throughout this manuscript that the investor's utility function $U$ is \emph{exponential}
\begin{align}
U(w)
	&=	-\frac{1}{\gam}\ee^{-\gam w} , &
\gam
	&>	0 , \label{eq:U.exp}
\end{align}
where $\gam$ is known as the \emph{risk aversion parameter}.
\end{assumption}
\noindent
Returning to PDE \eqref{eq:hjb.3} we now consider two cases: $\nu=0$ and $\nu \neq 0$.  We make the following ansatz:
\begin{align}
V(t,x,y,w,0)
	&=	-\frac{1}{\gam} \exp \Big( -\gam w + \eta(t,x,y) \Big) , \label{eq:ansatz.0} \\
V(t,x,y,w,\nu)
	&=	-\frac{1}{\gam} \exp \Big( -\gam w + \psi(t,x,y,\nu) \Big) . \label{eq:ansatz}
\end{align}
Clearly, $\eta$ can always be obtained from $\psi$ by setting $\nu=0$.  However, it will be useful to consider the special case $\nu=0$ separately.  Inserting \eqref{eq:ansatz.0} and \eqref{eq:ansatz} into \eqref{eq:hjb.3} we find the functions $\eta$ and $\psi$ satisfy
\begin{align}
0
	&=	\( \d_t + \Act \) \eta + \Bc (\eta) , &
\eta(T,x,y)
	&=	0 , \label{eq:hjb.5a} \\
0
	&=	\( \d_t + \Act \) \psi + \Bc (\psi) , &
\psi(T,x,y,\nu)
	&=	-\gam \nu \varphi(x,y) , \label{eq:hjb.5}  
\end{align}
where the linear operator $\Act$ and the nonlinear operator $\Bc$ are given by
\begin{align}
\Act
	&=	\tfrac{1}{2} \sig^2 \( \d_x^2 - \d_x \)
			+ \( c - \rho \beta \lam \) \d_y + \tfrac{1}{2} \beta^2 \d_y^2 + \rho \sig \beta \d_x \d_y  , &
\Bc(\psi)	
	&=	(1- \rho^2) (\tfrac{1}{2} \beta^2) (\d_y \psi)^2 - \tfrac{1}{2}\lam^2 .
\end{align}
From Definition \ref{def:indiff.price} and equations \eqref{eq:ansatz.0} and \eqref{eq:ansatz} we observe that the indifference price will be given by
\begin{align}
u(t,x,y,\nu)
	&=	\frac{1}{\gam \nu} \Big( \eta(t,x,y) - \psi(t,x,y,\nu) \Big) , \label{eq:u=eta.phi}
\end{align}
where we have removed the argument $w$ from the function $u$ as it is now clear this variable plays no role.  
One approach to finding the indifference price $u$ is to solve \eqref{eq:hjb.5a} and \eqref{eq:hjb.5} and then insert the solutions into \eqref{eq:u=eta.phi}.  Alternatively, one can seek a PDE for $u$ directly; both approaches will turn out to be useful.  Subtracting equation \eqref{eq:hjb.5} from \eqref{eq:hjb.5a} and dividing by $\gam \nu$, it is easy to show that the function $u$ satisfies
\begin{align}
0	
	&=	\( \d_t + \Act \) u + \Cc(u,\eta) , &
u(T,x,y,\nu)
	&=	\varphi(x,y) , \label{eq:hjb.5b} \\
\Cc(u,\eta)
	&:=	(1- \rho^2) (\tfrac{1}{2} \beta^2) \Big( 2 (\d_y u) (\d_y \eta ) - \gam \nu (\d_y u )^2 \Big) . 
\end{align}
Note that the PDE for $u$ depends on $\eta$, the solution of \eqref{eq:hjb.5a}.  Thus, to find and expression for $u$ we must solve a system of coupled nonlinear PDEs.
Note also that, upon finding expressions for $u$ and $\eta$, we can obtain $\psi$ from \eqref{eq:u=eta.phi}.
\par
Let us take a moment to discuss some probabilistic interpretations of the above PDEs.
\begin{remark}[On the minimal martingale measure]
\label{rmk:mmm}
Observe that $\Act$ is the generator of a process $(X,Y)$ whose dynamics under a probability measure $\Pbt$ are described by the following SDE
\begin{align}
\begin{aligned}
\dd X_t
	&=	-\tfrac{1}{2} \sig^2(X_t,Y_t) \dd t + \sig(X_t,Y_t) \dd \Bt_t^X , \\
\dd Y_t
	&=	\Big( c(X_t,Y_t) - \rho \beta(X_t,Y_t) \lam(X_t,Y_t) \Big) \dd t + \beta(X_t,Y_t) 
			\Big( \rho  \dd \Bt_t^X + \sqrt{1-\rho^2}  \dd \Bt_t^Y \Big) , 
\end{aligned} \label{eq:risk.neutral}
\end{align}
where $\Bt^X$ and $\Bt^Y$ are independent Brownian motions under $\Pbt$.  Dynamics \eqref{eq:risk.neutral} can be obtained from \eqref{eq:physical} under a Girsanov change of measure
\begin{align}
\frac{\dd \Pbt}{\dd \Pb} \Big|_{\Fc_t}
	&=	\exp \( - \frac{1}{2} \int_0^t \lam^2(X_s,Y_s) \dd s - \int_0^t \lam(X_s,Y_s) \dd B_s^X \) =: \widetilde{Z}_t , \label{eq:girsanov}
\end{align}
where we assume
\begin{align}
\Eb \int_0^T \lam^2(X_s,Y_s) \widetilde{Z}_s^2 \dd s
	&< \infty . \label{eq:integrability1}
\end{align}
Thus, $\Pbt$ is the \emph{minimal martingale} measure, as defined in \cite{follmer1991hedging}.
\end{remark}
\begin{remark}
Note that PDE \eqref{eq:hjb.5b} can alternatively be written as
\begin{align}
0	
	&=	( \d_t + \Ach ) u  , &
u(T,x,y,n)
	&=	\varphi(x) , \label{eq:u.alt}
\end{align}
where we have introduced linear operator $\Ach$ by
\begin{align}
\Ach
	&:=	\Act + \sqrt{1-\rho^2} \beta \Om  \d_y , &
\Om
	&:=	\sqrt{1-\rho^2} (\tfrac{1}{2} \beta) (\d_y \eta + \d_y \psi ) .
\end{align}
Observe that $\Ach$ is the generator of a process $(X,Y)$ whose dynamics under a probability measure $\Pbh$ are described by the following SDE
\begin{align}
\begin{aligned}
\dd X_t
	&=	-\tfrac{1}{2} \sig(X_t,Y_t) \dd t + \sig(X_t,Y_t) \dd \Bh_t^X , \\
\dd Y_t
	&=	\Big( c(X_t,Y_t) - \rho \beta(X_t,Y_t) \lam(X_t,Y_t) + \sqrt{1-\rho^2} \beta(X_t,Y_t) \Om(t,X_t,Y_t) \Big) \dd t \\ &\quad
			+ \beta(X_t,Y_t)  \Big( \rho  \dd \Bh_t^X + \sqrt{1-\rho^2}  \dd \Bh_t^Y \Big) , 
\end{aligned} \label{eq:risk.neutral.2}
\end{align}
where $\Bh^X$ and $\Bh^Y$ are independent Brownian motions under $\Pbh$.  Dynamics \eqref{eq:risk.neutral.2} can be obtained from \eqref{eq:physical} under a Girsanov change of measure
\begin{align}
\frac{\dd \Pbh}{\dd \Pb} \Big|_{\Fc_t}
	&=	\exp \( - \frac{1}{2} \int_0^t \( \lam^2(X_s,Y_s) + \Om^2(s,X_s,Y_s) \)\dd s 
			- \int_0^t \lam(X_s,Y_s) \dd B_s^X 
			+ \int_0^t \Om(X_s,Y_s) \dd B_s^Y \)  \\
	&=: \widehat{Z_t} ,  \label{eq:girsanov.2}
\end{align}
where we assume
\begin{align}
\Eb \int_0^T \( \lam^2(X_s,Y_s) + \Om^2(s,X_s,Y_s) \) \widehat{Z}_s^2 \dd s
	&< \infty . \label{eq:integrability2}
\end{align}
If, furthermore, we have that $u(\cdot,\cdot,\cdot,\nu) \in C^{1,2}([0,T) \times \Rb^2)$ and
\begin{align}
\Ebh \int_0^T \( \sig^2(X_s,Y_s) \( \d_x u(s,X_s,Y_s,\nu) \)^2 + \beta^2(X_s,Y_s) \( \d_y u(s,X_s,Y_s,\nu) \)^2 \) \dd s
	&<	\infty , \label{eq:integrability3}
\end{align}
then, using the Feynman-Kac representation, the solution $u$ of PDE \eqref{eq:u.alt} can be written as
\begin{align}
u(t,x,y,\nu)
	&=	\Ebh_{t,x,y} \, \varphi(X_T) , \label{eq:fk}
\end{align}
where $\Ebh$ denotes expectation under the probability measure $\Pbh$.  Keep in mind that, while PDE \eqref{eq:u.alt} is a linear PDE for $u$, we have not succeeded in removing the non-linearity from the indifference pricing problem since $\Ach$ depends on $\eta$ and $\psi$, which are solutions of nonlinear PDEs \eqref{eq:hjb.5a} and \eqref{eq:hjb.5}, respectively.
\end{remark}

\begin{remark}
\label{rmk:bsde}
An alternative approach to indifference pricing in an SV setting was recently introduced in \cite{sircar-sturm-2012} (see also \cite{kumar2015effect}).  Specifically, the authors compute indifference price $u_t^{(\nu)}$ as the difference of the solutions of two backward stochastic differential equations (BSDEs).  In our setting, this becomes
\begin{align}
u_t^{(\nu)}
	&=	\frac{1}{\gam \nu} \( R_t - R_t^{(\nu)} \) , \\
\dd R_t
	&=	- f(Z_t^{(X,0)},Z_t^{(Y,0)}) - Z_t^{(X,0)} \dd B_t^X -  Z_t^{(Y,0)} \dd B_t^Y , &
R_T
	&=	0 , \\
\dd R_t^{(\nu)}
	&=	- f(Z_t^{(X,\nu)},Z_t^{(Y,\nu)}) - Z_t^{(X,\nu)} \dd B_t^X -  Z_t^{(Y,\nu)} \dd B_t^Y , &
R_T^{(\nu)}
	&=	- \gam \nu \varphi(X_T,Y_T) ,
\end{align}
where $f$ is a strictly quadratic driver.
The authors then prove a generalized Feynman-Kac formula, which enables them to write 
\begin{align}
R_t
	&=	\widetilde{\eta}(t,X_t,Y_t) , &
R_t^{(\nu)}
	&=	\widetilde{\psi}(t,X_t,Y_t,\nu) ,
\end{align}
where $\widetilde{\eta}$ and $\widetilde{\psi}$ solve quasilinear parabolic PDEs.
When the driver $f$ belongs to a class of functions known as \emph{distorted entropic risk measures}
then the PDEs satisfied by $\widetilde{\eta}$ and $\widetilde{\psi}$
are of exactly the same form as \eqref{eq:hjb.5a} and \eqref{eq:hjb.5}, respectively
(the only nonlinearity arises from the $(\d_y \widetilde{\eta})^2$ and $(\d_y \widetilde{\psi})^2$ terms).
Even when restricting to drivers of the distorted entropic  risk measure variety, the BSDE approach is more general than the HJB approach we follow in this paper, as the distorted entropic risk measures depend on two parameters while the exponential utility we consider depends only on the risk-aversion parameter $\gam$.
However, from a computational point of view, the difficulty of solving two nonlinear PDEs remains the same in both approaches.
\end{remark}

Returning to the problem at hand, our goal is to find the indifference price $u$ of an option with payoff $\varphi(X_T,Y_T)$.  This requires either (i) finding a solution $\eta$ of PDE \eqref{eq:hjb.5a} and then solving PDE \eqref{eq:hjb.5b} for $u$, or (ii) finding a solution finding a solution $\eta$ of PDE \eqref{eq:hjb.5a} and a solution $\psi$ of \eqref{eq:hjb.5} and then obtaining $u$ from \eqref{eq:u=eta.phi}.  When $\varphi$ is a function of $x$ only, the PDEs that determine the indifference price cannot be linearized; in this case it will be best to follow method (i).  When $\varphi$ is a function of $y$ only, the PDEs that determine the indifference price can be linearized; in this case it will be best to follow method (ii).  Because the analysis of PDEs \eqref{eq:hjb.5a}, \eqref{eq:hjb.5} and \eqref{eq:hjb.5b} differs significantly depending on whether the payoff function $\varphi$ is a function of $x$ or a function of $y$, we will treat these cases separately.  We  begin in Section \ref{sec:asymptotics} with the case in which $\varphi$ is a function of $x$ only.  The case in which $\varphi$ that are a function of $y$ only will be analyzed in Section \ref{sec:Y}.  We will not consider the case in which $\varphi$ is a function of $(x,y)$ jointly.

%
%

\section{Options on traded assets}
\label{sec:asymptotics}

In this section we consider a European-style claim whose payoff function $\varphi$ depends only on the terminal value $X_T$ of a traded asset.  Thus, we make the following standing assumption.

\begin{assumption}
\label{ass:x}
Throughout Section \ref{sec:asymptotics} we assume the payoff function $\varphi$ is a function of $x$ only.  
\end{assumption}

We will divide Section \ref{sec:asymptotics} into four parts.
In Section \ref{sec:psi.asymptotics}, we will formally derive a sequence of PDEs which, if solved, will yield approximate indifference prices.
In Section \ref{sec:explicit}, we will give explicit solutions for this sequence of PDEs.
In Section \ref{sec:accuracy.X}, we will discuss the accuracy of our indifference pricing approximation.
And in Section \ref{sec:impvol} we will translate our indifference price approximation into an explicit approximation of implied volatility.

\subsection{PDE asymptotics}
\label{sec:psi.asymptotics}
As mentioned above, when $\varphi$ is a function of $x$ only, the preferred method for obtaining indifference prices is to find a solution $\eta$ of PDE \eqref{eq:hjb.5a} and then solve PDE \eqref{eq:hjb.5b} for $u$.
Since, for general $(\sig,c,\beta,\lam)$, there is no closed-form solution to \eqref{eq:hjb.5a} or \eqref{eq:hjb.5b}, we shall seek asymptotic solutions for these PDEs.  The approach we follow is similar to the approach taken in \cite{lorig-pagliarani-pascucci-2}, who, in a general LSV setting, obtain explicit approximations for (risk-neutral) European option prices and implied volatilities by expanding the coefficients of the underlying diffusion as a Taylor series.  Note that, because \eqref{eq:hjb.5a} and \eqref{eq:hjb.5b} are non-linear PDEs, solving these equations requires overcoming difficulties that are not present when solving the linear PDEs associated with risk-neutral pricing of European-style options.
\par
To begin the asymptotic analysis, let $\chi$ be any of the coefficients functions appearing in the operators $\Act$, $\Bc(\cdot)$ or $\Cc(\cdot,\cdot)$.  That is 
\begin{align}
\chi
	&\in	\{ (\tfrac{1}{2} \sig^2) , ( c - \rho \beta \lam ),  (\tfrac{1}{2} \beta^2), (\rho \sig \beta),  (\tfrac{1}{2} \lam^2) \} .
\end{align}
Fix a point $(\xb,\yb) \in \Rb^2$ and define
\begin{align}
\chi^\eps(x,y)
	&=	\chi\(\xb+\eps(x-\xb), \yb+\eps(y-\yb)\) , &
\eps
	&\in [0,1] . \label{eq:new.chi}
\end{align}
Consider now the following family of coupled PDEs indexed by $\eps$:
\begin{align}
0
	&=	(\d_t + \Act^\eps)\eta^\eps + \Bc^\eps(\eta^\eps) , &
\eta^\eps(T,x,y)
	&=	0 , &
\eps
	&\in [0,1] , \label{eq:psi.eps.pde}  \\
0
	&=	(\d_t + \Act^\eps)u^\eps + \Cc^\eps(u^\eps,\eta^\eps) , &
u^\eps(T,x,y,\nu)
	&=	\varphi(x) , &
\eps
	&\in [0,1] , \label{eq:eta.eps.pde}
\end{align}
where $\Act^\eps$, $\Bc^\eps(\cdot)$ and $\Cc^\eps(\cdot,\cdot)$ are obtained from $\Act$, $\Bc(\cdot)$ and $\Cc(\cdot,\cdot)$ by replacing the coefficients in these operators with their $\eps$-counterparts
\begin{align}
\Act^\eps
	&=	(\tfrac{1}{2} \sig^2)^\eps \( \d_x^2 - \d_x \)
			+ \( c - \rho \beta \lam \)^\eps \d_y + (\tfrac{1}{2} \beta^2)^\eps \d_y^2 
			+ (\rho \sig \beta)^\eps \d_x \d_y  , \\
\Bc^\eps(\eta)	
	&=	(1- \rho^2) ( \tfrac{1}{2} \beta^2)^\eps (\d_y \eta)^2 - (\tfrac{1}{2}\lam^2)^\eps , \\
\Cc^\eps(u,\eta)	
	&=	(1- \rho^2) ( \tfrac{1}{2} \beta^2)^\eps \Big( 2 (\d_y u )(\d_y \eta) - \gam \nu \(\d_y u \)^2 \Big) .
\end{align}

\begin{remark}
\label{rmk:eps=1}
Note that if we take $\eps = 1$, then PDEs \eqref{eq:psi.eps.pde} and \eqref{eq:eta.eps.pde} reduce to \eqref{eq:hjb.5a} and \eqref{eq:hjb.5b}, respectively.
Thus, the relationship between $(\eta^\eps,u^\eps)$ and $(\eta,u)$ is simply $(\eta^\eps,u^\eps)|_{\eps=1}=(\eta,u)$.
\end{remark}

We shall seek asymptotic solutions to \eqref{eq:psi.eps.pde} and \eqref{eq:eta.eps.pde} by expanding $\eta^\eps$ and $u^\eps$ in powers of $\eps$:
\begin{align}
\eta^\eps
	&=	\sum_{i=0}^\infty \eps^i \eta_i , &
u^\eps
	&=	\sum_{i=0}^\infty \eps^i u_i . \label{eq:u.expand} 
\end{align}
Our asymptotic solution to \eqref{eq:hjb.5a} and \eqref{eq:hjb.5b}, which is the case we are interested in, will then follow by taking $\eps=1$. 
Since the coefficients $\{\chi^\eps\}$ appearing in the operators $\Act^\eps$, $\Bc^\eps(\cdot)$ and $\Cc^\eps(\cdot,\cdot)$ depend on $\eps$, we formally expand these as well
\begin{align}
\chi^\eps(x,y)
	&:=	\sum_{n=0}^\infty \eps^n \chi_n(x,y) , &
\eps 
	&\in [0,1] , \label{eq:chi.expand} \\
\chi_n(x,y)
	&:=	\sum_{k=0}^n \chi_{n-k,k} \cdot (x-\xb)^{n-k}(y-\yb)^k , &
\chi_{n-k,k}
	&:=	\frac{1}{(n-k)!k!} \d_x^{n-k}\d_y^k \chi(\xb,\yb) . \label{eq:chi.n}
\end{align}
Here, we assume for simplicity that the coefficients $\{ \chi^\eps \}$ are entire functions and therefore equal to their power series expansions.  In fact, we shall see that our $m$th order approximation of the indifference price requires only $\{\chi^\eps\} \in C^m(\Rb^2)$.

\begin{remark}
\label{rmk:idea}
The idea behind our expansion can be understood as follows.  The parameter $\eps$ controls how quickly the coefficients $\{\chi^\eps(x,y)\}$ vary with $(x,y)$.  For example, from \eqref{eq:new.chi} we see that $\chi^\eps(x,y)|_{\eps = 0} = \chi(\xb,\yb)$, which is a constant since $(\xb,\yb)$ is fixed.  As we dial up the value of $\eps$ from $0$ to $1$ we obtain a function $\chi^\eps(x,y)$ that varies more than a constant, but not more than $\chi^\eps(x,y)|_{\eps = 1} = \chi(x,y)$.  Figure \ref{fig:eps} illustrates this visually.  As we shall see, when $\eps = 0$, PDEs \eqref{eq:psi.eps.pde} and \eqref{eq:eta.eps.pde} can be solve explicitly.
The exact solutions in the $\eps=0$ case give us a foundation upon which we can construct approximate solutions when $\eps > 0$.  As long as the coefficients $\{\chi^\eps(x,y)\}$ do not vary too much, our approximate solutions for PDEs \eqref{eq:psi.eps.pde} and \eqref{eq:eta.eps.pde} will be close to the exact solutions -- even when $\eps = 1$.
In fact, when we discuss options with payoffs $\varphi(y)$ in Section \ref{sec:Y},
we will rigorously show that a higher degree of regularity of the coefficients $\{ \chi \}$ allows us to construct a higher order approximate solution (Proposition \ref{thm:xi}), which in turn, yields a more accurate approximate of indifference prices (Corollary \ref{cor:u.bound}).
\end{remark}

In order to find $\eta_i$ and $u_i$ $(i \geq 0)$, we insert the expansions for $\eta^\eps$ and $u^\eps$ into \eqref{eq:psi.eps.pde} and \eqref{eq:eta.eps.pde}, respectively, and collect terms of like powers in $\eps$.  At lowest order we obtain
\begin{align}
\Oc(1):&&
0
	&=	(\d_t + \Act_0) \eta_0 + (1- \rho^2) ( \tfrac{1}{2} \beta^2)_0 (\d_y \eta_0)^2 - (\tfrac{1}{2}\lam^2)_0 , &
\eta_0(T,x,y)
	&=	0 , \label{eq:psi0.pde} \\
\Oc(1):&&
0
	&=	(\d_t + \Act_0) u_0 + (1- \rho^2) ( \tfrac{1}{2} \beta^2)_0 \Big( 2 (\d_y u_0 )(\d_y \eta_0) - \gam \nu \(\d_y u_0 \)^2 \Big) , &
u_0(T,x,y,\nu)
	&=	\varphi(x) , \label{eq:eta0.pde}
\end{align}
where we have defined
\begin{align}
\Act_n
	&:=	(\tfrac{1}{2} \sig^2)_n \( \d_x^2 - \d_x \)
			+ ( c - \rho \beta \lam )_n \d_y + (\tfrac{1}{2} \beta^2)_n \d_y^2 
			+ (\rho \sig \beta)_n \d_x \d_y , &
n
	&\in	\{ 0 \} \cup \Nb . \label{eq:An}
\end{align}
Observe that the coefficients of $\Act_0$ as well as $(\tfrac{1}{2}\lam^2)_0$ are constants since, from \eqref{eq:chi.n}, we have $\chi_0 := \chi(\xb,\yb)$.  Moreover, the terminal condition for $\eta_0$ does not depend on $(x,y)$ and the terminal condition for $u_0$ does not depend on $y$.  These observations suggest that we seek a solution $\eta_0$ of \eqref{eq:psi0.pde} that is a function of $t$ only and a solution $u_0$ of \eqref{eq:eta0.pde} that is a function of $(t,x)$ only.  If we do this, equations \eqref{eq:psi0.pde} and \eqref{eq:eta0.pde} become
\begin{align}
\Oc(1):&&
0
	&=	\d_t \eta_0 - (\tfrac{1}{2}\lam^2)_0 , &
\eta_0(T)
	&=	0 , \label{eq:psi0.pde.2} \\
\Oc(1):&&
0
	&=	(\d_t + \Act_0) u_0 , &
u_0(T,x)
	&=	\varphi(x) . \label{eq:eta0.pde.2}
\end{align}
As we shall see in Section \ref{sec:explicit}, ODE \eqref{eq:psi0.pde.2} and PDE \eqref{eq:eta0.pde.2} can be solved explicitly.
Thus, we can identify the solution $\eta_0$ of \eqref{eq:psi0.pde.2} and the solution $u_0$ of \eqref{eq:eta0.pde.2} as the 
solutions of \eqref{eq:psi0.pde} and \eqref{eq:eta0.pde}, respectively.
\par
Next, collecting terms of higher order in $\eps$, it is straightforward to show that the $\Oc(\eps^m)$ equations are of the form
\begin{align}
\Oc(\eps^m):&&
0
	&=	(\d_t + \Act_0) \eta_m + H_m , &
\eta_m(T,x,y)
	&=	0 , \label{eq:psin.pde} \\
\Oc(\eps^m):&&
0
	&=	(\d_t + \Act_0) u_m + U_m , &
u_m(T,x,y,\nu)
	&=	0 , \label{eq:etan.pde}
\end{align}
where the source terms $H_m$ and $U_m$ are given by
\begin{align}
\Oc(\eps^m):&&
H_m
	&=	\sum_{k=1}^m \Act_k \eta_{m-k} - (\tfrac{1}{2} \lam^2)_m 
			+ (1-\rho^2) \sum_{k,i,j \in K_m} (\tfrac{1}{2} \beta^2)_k (\d_y \eta_i ) (\d_y \eta_j), \label{eq:Hm} \\
\Oc(\eps^m):&&
U_m
	&=	\sum_{k=1}^m \Act_k \eta_{m-k} 
			+ (1-\rho^2) \sum_{k,i,j \in K_m} (\tfrac{1}{2} \beta^2)_k \Big( 2 (\d_y u_i ) (\d_y \eta_j) - \gam \nu (\d_y u_i ) (\d_y u_j) \Big), 
			\label{eq:Um} \\ &&
K_m
	&=	\{(i,k,j) \in \Nb_0^3 : i+j+k=m \,\,\, \text{and} \,\,\, i,j,k \neq m \} .
\end{align}
The reason the set $K_m$ does not include $\Oc(\eps^m)$ terms is because if such a term appears, then it would be multiplied in \eqref{eq:Hm} and \eqref{eq:Um} by either $(\d_y u_0)$ or $(\d_y \eta_0)$, both of which are zero (since $\eta_0$ and $u_0$ are independent of $y$).
Explicitly, the $\Oc(\eps)$ source terms $H_1$ and $U_1$ are given by
\begin{align}
\Oc(\eps):&&
H_1
	&=	\Act_1 \eta_0 - (\tfrac{1}{2}\lam^2)_1 , \label{eq:F1} \\
\Oc(\eps):&&
U_1
	&=	\Act_1 u_0 . \label{eq:G1} 
\end{align}
and the $\Oc(\eps^2)$ source terms $H_2$ and $U_2$ are 
\begin{align}
\Oc(\eps^2):&&
H_2
	&=	\Act_2 \eta_0 + \Act_1 \eta_1 - (\tfrac{1}{2}\lam^2)_2  + (1-\rho^2) ( \tfrac{1}{2} \beta^2)_0 (\d_y \eta_1)^2 , \label{eq:F2} \\
\Oc(\eps^2):&&
U_2
	&=	\Act_2 u_0 + \Act_1 u_1 +	(1- \rho^2) ( \tfrac{1}{2} \beta^2)_0 \Big( 2 (\d_y u_1 )(\d_y \eta_1) - \gam \nu (\d_y u_1 )^2 \Big) 
			\label{eq:G2} , 
\end{align}
\begin{remark}
\label{rmk:H2.H1}
Note, to obtain second order term $u_2$ in the indifference price expansion, one does not require $\eta_2$. 
However, for completeness, we present the source term $H_2$ along with $H_1$.
\end{remark}
\noindent
This is as far as we will take the asymptotic analysis of PDEs \eqref{eq:psi.eps.pde} and \eqref{eq:eta.eps.pde}.  Having served its purpose, we set $\eps=1$ and make the following definition:
\begin{definition}
\label{def:ub.n}
Let $\eta$ and $u$ be the solutions of \eqref{eq:hjb.5a} and \eqref{eq:hjb.5b}, respectively.  
Assume $\varphi$ is a function of $x$ only (Assumption \ref{ass:x}) and the coefficients $(\tfrac{1}{2}\lam^2)$, $(\tfrac{1}{2} \sig^2)$, $( c - \rho \beta \lam )$, $(\tfrac{1}{2} \beta^2)$ and $(\rho \sig \beta)$ are $C^m(\Rb^2)$.  Then the \emph{$m$-th order approximations} of $\eta$ and $u$ are defined as
\begin{align}
\etab_m
	&:=	\sum_{i=0}^m \eta_i , &
\ub_m
	&:=	\sum_{i=0}^m  u_i , \label{eq:approximations}
\end{align}
where $\eta_0$ and $u_0$ solve \eqref{eq:psi0.pde.2} and \eqref{eq:eta0.pde.2}, respectfully,
and $\eta_i$ and $u_i$ $(i \geq 1)$ solve \eqref{eq:psin.pde} and \eqref{eq:etan.pde}, respectfully.
\end{definition}
\noindent
In Section \ref{sec:explicit} we will derive explicit (non-integral) expressions for $\eta_i$ and $u_i$ $(i \leq 2)$ and explicit integral expressions for $\eta_i$ and $u_i$ $(i \geq 0)$.  First, we make the following observation:
\begin{remark}
The effect of the nonlinear term $\Cc(u,\eta)$ in \eqref{eq:hjb.5b} has no effect at zeroth and first order and is felt for the first time at second order.  To see this, let us define
\begin{align}
q(t,x,y)
	&:=		\Ebt_{t,x,y} \, \varphi(X_T)  , \label{eq:q.def}
\end{align}
where $\Ebt$ denotes expectation under $\Pbt$, defined in \eqref{eq:girsanov}.
Note that $q$ can be interpreted the price of a European-style option assuming $(X,Y)$ has risk-neutral dynamics given by \eqref{eq:risk.neutral}.  The function $q$ satisfies the \emph{linear} pricing PDE
\begin{align}
0
	&=	(\d_t + \Act) q , &
q(T,x,y)
	&=	\varphi(x) . \label{eq:q.pde}
\end{align}
Repeating step-by-step the asymptotics analysis above, one could replace $\Act$  in \eqref{eq:q.pde} with $\Act^\eps$ and seek a solution $q^\eps = \sum_{k=0}^\infty \eps^k q_k$.  Upon collecting terms of like order of $\eps$ one would find
\begin{align}
\Oc(1):&&
0
	&=	(\d_t + \Act_0) q_0 , &
q_0(T,x)
	&=	\varphi(x) , \label{eq:q0.pde}  \\
\Oc(\eps^m):&&
0
	&=	(\d_t + \Act_0) q_m + Q_m, &
q_m(T,x,y)
	&=	0 , &
m
	&\geq 1 , \label{eq:qm.pde}
\end{align}
where the $m$th-order source term $Q_m$ is given by
\begin{align}
\Oc(\eps^m):&&
Q_m
	&=	\sum_{k=1}^m \Act_k q_{m-k} . \label{eq:Qm}
\end{align}
Comparing the PDE for $u_0$ \eqref{eq:eta0.pde.2} with the PDE for $q_0$ \eqref{eq:q0.pde}, the PDE for $u_m$ \eqref{eq:etan.pde} with the PDE for $q_m$ \eqref{eq:qm.pde} and source terms $U_1$ \eqref{eq:G1} and $U_2$ \eqref{eq:G2} with the source term $Q_m$ \eqref{eq:Qm}, we see that
\begin{align}
u_0 
	&=	q_0 , &
u_1	
	&=	q_1 , &
u_2 
	&=	q_2 + u_2^\Ind , \label{eq:u.and.q}
\end{align}
where $u_2^\Ind$ results from the nonlinear term $\Cc(u,\eta)$ in \eqref{eq:hjb.5b} and satisfies
\begin{align}
0
	&=	(\d_t + \Act_0) u_2^\Ind + U_2^\Ind , &
U_2^\Ind
	&=	(1- \rho^2) ( \tfrac{1}{2} \beta^2)_0 \Big( 2 (\d_y u_1 )(\d_y \eta_1) - \gam \nu (\d_y u_1 )^2 \Big) . \label{eq:u2.indiff}
\end{align}
and  terminal condition $u_2^\Ind(T,x,y)=0$.  Thus, the second order approximation of the indifference price is given by $\ub_2 = \qb_2 + u_2^\Ind$ where $\qb_2 : = q_0 + q_1 + q_2$ arises from a (standard) \emph{linear} pricing rule \eqref{eq:q.def} and $u_2^\Ind$ is a correction that arises from the nonlinear aspect of indifference pricing.  
\end{remark}

\subsection{Explicit expressions}
\label{sec:explicit}
In this section, we derive general integral expressions for $\eta_i$ and $u_i$ $(i \geq 0)$ and explicit (non-integral) expressions for the functions $\eta_0$, $\eta_1$ and $\eta_2$ as well as $u_0$, $u_1$ and $u_2$.  
To begin, we observe that the linear operator $\Act_0$ 
is the infinitesimal generator of a diffusion in $\Rb^2$ whose drift vector and covariance matrix are constant
(i.e., a correlated Brownian motion with drift).  
The \emph{semigroup} $\Pc_0(t,t_1)$ generated by $\Act_0$ is given by
\begin{align}
\Pc_0(t,t_1) \eta(x,y)
	&:=	\int_{\Rb^2} \dd x_1 \dd y_1  \Gam_0(t,x,y;t_1,x_1,y_1) \eta(x_1,y_1) , &
t_1
	&\geq t , \label{eq:Pc0}
\end{align}
where the function $\Gam_0$ is the \emph{fundamental solution} corresponding to the linear operator $(\d_t + \Act_0)$.  It is easy to show that $\Gam_0$ is a Gaussian kernel
\begin{align}
\Gam_0(t,x,y;t_1,x_1,y_1)
    &=  \frac{1}{\sqrt{( 2\pi )^2|\Cv|} } \exp\( -\frac{1}{2}\mv^\text{T} \Cv^{-1}\mv \) , 
				\label{eq:Gam.0}
\end{align}
where the covariance matrix $\Cv$ and vector $\mv$ are given by
\begin{align}
\Cv
	&=	(t_1-t) \begin{pmatrix}
      ( \sig^2 )_0 & ( \rho \sig \beta )_0 \\
      ( \rho \sig \beta )_0 & ( \beta^2 )_0 
      \end{pmatrix} , &
\mv
 &= \begin{pmatrix}
    x_1 - x - (t_1-t) (-\frac{1}{2}\sig^2 )_0 \\
		y_1 - y - (t_1-t) \( c - \rho \beta \lam \)_0 \
    \end{pmatrix}  .
\end{align}
Note that $\Pc_0(t,t_1)$ enjoys the \emph{semigroup property}:
\begin{align}
\Pc_0(t,t_1) \Pc_0(t_1,t_2)
	&=	\Pc_0(t,t_2) , & 
t
	&\leq t_1 \leq t_2 . \label{eq:semigroup}
\end{align}
By Duhamel's principle, the unique classical solution (if it exists) to any PDE of the form
\begin{align}
0
	&=	(\d_t + \Act_0) v + F , &
v(T,x,y)
	&=	G(x,y) , \label{eq:form}
\end{align}
is given by, 
\begin{align}
v(t)
	&=	\Pc_0(t,T) G + \int_t^T \dd t_1  \Pc_0(t,t_1) F(t_1) , \label{eq:duhamel}
\end{align}
where, for simplicity, we have omitted the arguments $(x,y)$.  As PDEs \eqref{eq:psi0.pde.2}, \eqref{eq:eta0.pde.2}, \eqref{eq:psin.pde} and \eqref{eq:etan.pde} are all of the form \eqref{eq:form}, one can in principle use \eqref{eq:duhamel} to compute explicit expressions for each $\eta_i$ and $u_i$ $(i \geq 0)$.
Observe, however, that computing the right-hand side of \eqref{eq:duhamel} requires evaluating a double integral in the spatial variables (because the semigroup operator $\Pc_0(t,s)$, given by \eqref{eq:Pc0}, is an integral operator) as well as an integral in the temporal variable.  
\par
It will be helpful to find more explicit expressions for the functions $\eta_0$, $\eta_1$ and $\eta_2$ as well as $u_0$, $u_1$ and $u_2$.
To this end, we introduce the following operators
\begin{align}
\Xc(t,t_1)
	&=	 x + (t_1-t) \Big( -( \tfrac{1}{2}\sig^2 )_0 + 2 (\tfrac{1}{2}\sig^2)_0 \d_x
					+ (\rho \sig \beta)_0 \d_y \Big) , & t_1 &\geq t , \label{eq:X} \\
\Yc(t,t_1)
	&=	 y + (t_1-t) \Big( (c-\rho \beta \lam)_0 + 2 (\tfrac{1}{2}\beta^2)_0 \d_y
					+ (\rho \sig \beta)_0 \d_x  \Big) , & t_1 &\geq t ,  \label{eq:Y}
\end{align}
It is easy to check by direct computation that the operators $\Xc(t,t_1)$ and $\Yc(t,t_1)$ commute and have the following property
\begin{align}
\(\Xc(t,t_1) \)^n \(\Yc(t,t_1)\)^m \Gam_0(t,x,y,t_1,x_1,y_1)
	&=	\( x_1 \)^n \( y_1 \)^m \Gam_0(t,x,y,t_1,x_1,y_1) , &
n,m
	&\in \mathbb{N}_0 . \label{eq:prop}
\end{align}
Therefore, if $f$ is a polynomial function of $(x,y)$ we have
\begin{align}
\Pc_0(t,t_1) f(x,y)
	&=	\int_{\Rb^2} \dd x_1 \dd y_2 \, f(x_1,y_1) \Gam_0(t,x,y;t_1,x_1,y_1) \\
	&=	f(\Xc(t,t_1),\Yc(t,t_1)) \int_{\Rb^2} \dd x_1 \dd y_2 \, \Gam_0(t,x,y;t_1,x_1,y_1) \\
	&=	f(\Xc(t,,t_1),\Yc(t,t_1)) 1. \label{eq:Pc.poly}
\end{align}
It will also be helpful to introduce the following operator
\begin{align}
\Gc_n(t,t_1)
	&:=	\Act_n(\Xc(t,t_1),\Yc(t,t_1)) , & t_1 &\geq t , &
n
	&\geq 1 , \label{eq:Gc}
\end{align}
where the notation $\Act_n(\Xc(t,t_1),\Yc(t,t_1))$ indicates that the $(x,y)$-dependence in coefficients of $\Act_n \equiv \Act_n(x,y)$, has been replaced by $(x,y) \to (\Xc(t,t_1),\Yc(t,t_1))$.  For example, the term
\begin{align} 
(\tfrac{1}{2} \sig^2)_1 ( \d_x^2 - \d_x )
	&:= \Big( (\tfrac{1}{2} \sig^2)_{1,0} (x-\xb) + (\tfrac{1}{2} \sig^2)_{0,1}(y-\yb) \Big) ( \d_x^2 - \d_x ),
\end{align}
appearing in $\Act_1$ becomes
\begin{align}
\Big( (\tfrac{1}{2} \sig^2)_{1,0} (\Xc(t,t_1)-\xb) + (\tfrac{1}{2} \sig^2)_{0,1}(\Yc(t,t_1)-\yb) \Big) ( \d_x^2 - \d_x ) ,
\end{align}
in $\Gc_1(t,t_1):=\Act_1(\Xc(t,t_1),\Yc(t,t_1))$.  The following lemma will be essential for the computations that follow
\begin{lemma}
\label{lem:com}
Let the operators $\Act_n$, $\Pc_0(t,t_1)$, $\Xc(t,t_1)$, $\Yc(t,t_1)$ and $\Gc_n(t,t_1)$ be given by \eqref{eq:An}, \eqref{eq:Pc0}, \eqref{eq:X}, \eqref{eq:Y} and \eqref{eq:Gc} respectively.  Then we have the following commutation-like relation
\begin{align}
\Pc_0(t,t_1) \Act_n f
	&=	\Gc_n(t,t_1) \Pc_0(t,t_1) f , \label{eq:PA=GP}
\end{align}
which holds for any measurable function $f \in C^{n+2}(\Rb^2)$ whose partial derivatives of all orders less than or equal to $n+2$ are measurable and at most exponentially growing.
\end{lemma}
\begin{proof}
It is sufficient to prove the claim for $\Act_n$ of the form
\begin{align}
\Act_n
	&=	x^k y^m \d_x^i \d_y^j , &
k,m,i,j 
	&\in \Nb_0 , \label{eq:A.form}
\end{align}
since each $\Act_n$ is a finite sum of terms of this form.  Let $f$ be measurable and at most exponentially growing.  Then
\begin{align}
\Pc_0(t,t_1) \Act_n f(x,y)
	&=	\int_{\Rb^2} \dd x_1 \dd y_1 \, \Gam_0(t,x,y;t_1,x_1,y_1) x_1^k y_1^m \d_{x_1}^i \d_{y_1}^j f(x_1,y_1) \\
	&=	(\Xc(t,t_1))^k (\Yc(t,t_1))^m \int_{\Rb^2}  \dd x_1 \dd y_1 \, \Gam_0(t,x,y;t_1,x_1,y_1) \d_{x_1}^i \d_{y_1}^j f(x_1,y_1) \\
	&=	(\Xc(t,t_1))^k (\Yc(t,t_1))^m (-1)^{i+j} \int_{\Rb^2}  \dd x_1 \dd y_1 \, f(x_1,y_1) \d_{x_1}^i \d_{y_1}^j \Gam_0(t,x,y;t_1,x_1,y_1) \\
	&=	(\Xc(t,t_1))^k (\Yc(t,t_1))^m \d_{x}^i \d_{y}^j \int_{\Rb^2}  \dd x_1 \dd y_1 \, \Gam_0(t,x,y;t_1,x_1,y_1) f(x_1,y_1)  \\
	&=	\Gc_n(t,t_1) \Pc_0(t,t_1) f(x,y) . \label{eq:dxdy.added}
\end{align}
In the first equality we have simply used the definition of $\Pc_0(t,t_1)$ \eqref{eq:Pc0} and the form of $\Act_n$ \eqref{eq:A.form}.  In the second equality we have used \eqref{eq:prop} and pulled the operator $(\Xc(t,t_1))^k (\Yc(t,t_1))^m$ out of the integral.  In the third equality we have integrated by parts.  In the fourth equality we have used the fact that the Gaussian kernel $\Gam_0$ can be written as function of $(x-x_1)$ and $(y-y_1)$.  Therefore 
\begin{align}
\d_{x_1}^i \d_{y_1}^j \Gam_0(t,x,y;t_1,x_1,y_1) 
	&= (-1)^{i+j} \d_{x}^i \d_{y}^j \Gam_0(t,x,y;t_1,x_1,y_1) .
\end{align}
And in the last equality we have used the definitions of $\Pc_0(t,t_1)$ \eqref{eq:Pc0} and $\Gc_n$ \eqref{eq:Gc}.
\end{proof}

\noindent
Using the operators $\Xc(t,t_1)$, $\Yc(t,t_1)$, and $\Gc_i(t,t_1)$, it is now straightforward to establish the following proposition, which provides explicit expressions for $\eta_0$, $\eta_1$ and $\eta_2$.
\begin{proposition}
\label{thm:phi}
Let $\eta_0$ be the unique classical solution of \eqref{eq:psi0.pde.2} and $\eta_1$ and $\eta_2$ be the unique classical solutions of \eqref{eq:psin.pde} with source terms $H_1$ \eqref{eq:F1} and $H_2$ \eqref{eq:F2}, respectfully.  
Assume the coefficients $(\tfrac{1}{2}\lam^2)$, $(\tfrac{1}{2} \sig^2)$, $( c - \rho \beta \lam )$,  $(\tfrac{1}{2} \beta^2)$ and $(\rho \sig \beta)$ are $C^2(\Rb^2)$.
Then, omitting $(x,y)$-dependence for clarity, we have
\begin{align}
\eta_0(t)
	&= - (T-t) (\tfrac{1}{2}\lam^2)_0 , \label{eq:eta0.theorem} \\
\eta_1(t)
	&=	- \int_t^T \dd t_1  (\tfrac{1}{2}\lam^2)_1(\Xc(t,t_1),\Yc(t,t_1)) 1 , \label{eq:eta1.theorem} \\
\eta_2(t)
	&=	- \int_t^T \dd t_1 \int_{t_1}^T \dd t_2 \, \Gc_1(t,t_1) (\tfrac{1}{2}\lam^2)_1(\Xc(t,t_2),\Yc(t,t_2)) 1 \\ &\quad
			- \int_t^T \dd t_1  (\tfrac{1}{2}\lam^2)_2 (\Xc(t,t_1),\Yc(t,t_1)) 1
			+ \frac{1}{3}(T-t)^3 (1-\rho^2) ( \tfrac{1}{2} \beta^2)_0 (\tfrac{1}{2}\lam^2)_{0,1}^2 , \label{eq:eta2.theorem}
\end{align}
where the operators $\Xc(t,t_1)$, $\Yc(t,t_1)$ and $\Gc_i(t,t_1)$ are given by \eqref{eq:X}, \eqref{eq:Y} and \eqref{eq:Gc}, respectively, and
$(\tfrac{1}{2}\lam^2)_{0,1}=\d_y(\tfrac{1}{2}\lam^2(\xb,\yb))$, as described in \eqref{eq:chi.n}.
\end{proposition}
\begin{proof}
See Appendix \ref{sec:phi.proof}.
\end{proof}
\begin{remark}
\label{rmk:phi}
Note that $\eta_1$ is obtained by computing
\begin{align}
(\tfrac{1}{2}\lam^2)_1(\Xc(t,t_1),\Yc(t,t_1)) 1
	&=	(\tfrac{1}{2}\lam^2)_{1,0}(\Xc(t,t_1)-\xb) 1 + (\tfrac{1}{2}\lam^2)_{0,1}(\Yc(t,t_1)-\yb) 1 ,
\end{align}
and then integrating in with respect to $t_1$.  An explicit computation yields that $\eta_1$ is given by 
\begin{align}
\eta_1(t,x,y)
	&=	-(\tfrac{1}{2}\lam^2)_{1,0} \( (T-t) (x-\xb)-\frac{1}{2}(T-t)^2(\tfrac{1}{2}\sig^2 )_0 \) \\ &\quad
			-(\tfrac{1}{2}\lam^2)_{0,1} \( (T-t) (y-\yb)+\frac{1}{2}(T-t)^2(c-\rho \beta \lam )_0 \) . \label{eq:eta1}
\end{align}
The function $\eta_2$ is the sum of three terms, two of which are obtained by computing
\begin{align}
\Gc_1(t,t_1) (\tfrac{1}{2}\lam^2)_1(\Xc(t,t_2),\Yc(t,t_2)) 1
	&=	\Act_1(\Xc(t,t_1),\Yc(t,t_1))\,(\tfrac{1}{2}\lam^2)_1(\Xc(t,t_2),\Yc(t,t_2)) 1 , \\
(\tfrac{1}{2}\lam^2)_2 (\Xc(t,t_1),\Yc(t,t_1)) 1
	&=	\sum_{k=0}^2 (\tfrac{1}{2}\lam^2)_{k,2-k} (\Xc(t,t_1)-\xb)^k(\Yc(t,t_1)-\yb)^{2-k} 1 ,
\end{align}
and then integrating with respect to $t_1$.  For brevity, we omit the explicit expression for $\eta_2$.
\end{remark}

To compute an approximation for the indifference price $u$, we must specify a payoff function $\varphi$.
Since we are interested in characterizing the buyer and seller implied volatility surfaces, we focus here on call payoffs $\varphi(x) = (\ee^x - \ee^k)^+$.  Before presenting expressions for $u_0$, $u_1$ and $u_2$, it will be helpful to define the Black-Scholes call price.
\begin{definition}
For a fixed maturity date $T>t$ and $\log$ strike $k$, the \emph{Black-Scholes call price} $u^\BS$ is defined as
\begin{align}
u^\BS(t,x;\sig)
	&:=	\ee^x \Phi\(d_+(t,x;\sig)\) - \ee^k \Phi\(d_-(t,x;\sig)\) , \label{eq:uBS}
\end{align}
where $\Phi$ is a standard normal CDF and
\begin{align}
d_\pm(t,x;\sig)
	&=	\frac{1}{\sig\sqrt{T-t}} \( x-k \pm \frac{1}{2} \sig^2 (T-t) \) . \label{eq:d}
\end{align}
\end{definition}
\begin{proposition}
\label{thm:u}
Fix a maturity date $T$, a $\log$ strike $k$ and assume $\varphi$ is the payoff of a European call option $\varphi(x)=(\ee^x-\ee^k)^+$.
Let $u_0$ be the unique classical solution of \eqref{eq:eta0.pde.2} and let $u_1$ and $u_2$ be the unique classical solutions of \eqref{eq:etan.pde} with source terms $U_1$ \eqref{eq:G1} and $U_2$ \eqref{eq:G2}, respectfully.  
Assume that $(\tfrac{1}{2}\lam^2)$ is $C^1(\Rb^2)$ and $(\tfrac{1}{2} \sig^2)$, $( c - \rho \beta \lam )$,  $(\tfrac{1}{2} \beta^2)$ and $(\rho \sig \beta)$ are $C^2(\Rb^2)$.
Then, omitting arguments of $(x,y)$ for clarity, we have 
\begin{align}
u_0(t) 
	&=		u^\BS(t,\cdot;\sig_0) , \label{eq:u0.theorem} \\
u_1(t)
	&=	\Big( \int_t^T \dd t_1  \Gc_1(t,t_1) \Big) u^\BS(t,\cdot;\sig_0) , \label{eq:u1.theorem} \\
u_2(t)
	&=	u_2^0(t) + u_2^\Ind(t) , \label{eq:u2.theorem} \\
u_2^0(t)
	&=	\Big( \int_t^T \dd t_1  \Gc_2(t,t_1) + \int_t^T \dd t_1 \int_{t_1}^T \dd t_2  \Gc_1(t,t_1)\Gc_1(t,t_2) \Big) 
			u^\BS(t,\cdot;\sig_0) , \label{eq:u2.0} \\
u_2^\Ind(t)
	&=	(1 - \rho^2) (\tfrac{1}{2} \beta^2)_0 \bigg\{ 2
			\Big( \int_t^T \dd t_1 \int_{t_1}^T \dd t_2  (\d_y \eta_1(t_1)) \cdot \d_y \Gc_1(t,t_2) \Big) u^\BS(t,\cdot;\sig_0) \\ &\quad
			- \gam \nu  \frac{(\tfrac{1}{2}\sig^2)_{0,1}^2 }{2\pi\sig_0^2} 
			\int_t^T \dd t_1  \frac{(T-t_1)^{3/2} \ee^{2k}}{\sqrt{T-t+t_1-t}}
			\exp \(\frac{-\((k-x)+\frac{1}{2} \sig_0^2 (T-t)\)^2}{\sig_0^2 (T-t+t_1-t)} \) \bigg\}. \label{eq:u2.indiff.2}
\end{align}
where the  operator $\Gc_i(t,t_1)$ and the functions $\eta_1$ and $u^\BS$ are defined in \eqref{eq:Gc}, \eqref{eq:eta1} and \eqref{eq:uBS}, respectively,
\end{proposition}
\begin{proof}
See Appendix \ref{sec:u.proof}.
\end{proof}
\begin{remark}
In equations \eqref{eq:u1.theorem}, \eqref{eq:u2.0} and \eqref{eq:u2.indiff.2} the integrals with respect to $t_1$ and $t_2$ are performed on the operators that act on $u^\BS(t,\cdot;\sig_0)$ and \emph{not} on the function $u^\BS (t,\cdot;\sig_0)$.  These integrals can be computed explicitly.  As such, $u_1$ and $u_2$ are obtained simply by applying a differential operator to $u^\BS(t,\cdot,\sig_0)$.
\end{remark}
\begin{remark}
Observe that we have isolated the term $u_2^\Ind$, which arises from the source term $U^\Ind$ and satisfies PDE \eqref{eq:u2.indiff}.  Also observe that the second term in \eqref{eq:u2.indiff.2} has the sign of $(-\nu)$.  Thus, an investor that is looking to \emph{buy} $|\nu|$ call options will have a lower indifference price than an investor that is looking to \emph{sell} $|\nu|$ call options, giving rise to a bid-ask spread.  
\end{remark}
\begin{remark}
Note, if we restrict $X$ to local volatility dynamics (i.e., $\sig$ as a function of $x$ only), then we are in a complete market setting and $u^\Ind=0$ as it should be.  To see this observe that, when acting on a function independent of $y$ we have 
\begin{align}
\d_y \Gc_1(t,t_1)
	&=	(\tfrac{1}{2}\sig^2)_{0,1} (\d_x^2 - \d_x ) . \label{eq:dy.G1}
\end{align}
From \eqref{eq:chi.n}, we also have $(\tfrac{1}{2}\sig^2)_{0,1}=\d_y(\tfrac{1}{2}\sig^2(\xb))=0$.  Thus, both terms appearing in \eqref{eq:u2.indiff.2} are zero.  
\end{remark}

\subsection{Accuracy of the indifference price approximation}
\label{sec:accuracy.X}
Establishing the asymptotic accuracy of the solution of a nonlinear PDE is a notoriously difficult problem.  It is well beyond the scope of this paper to rigorously establish the accuracy of our second order indifference price approximation $\ub_2$.  Nevertheless, some discussion of the approximation's accuracy is necessary in order to establish conditions under which use of the approximation is appropriate.  In this section, we establish the accuracy of our zeroth order indifference price approximation $\ub_0$, and we outline the steps and the conditions that are needed to establish the accuracy of higher order approximations $\ub_m$ ($m \geq 1$).
\par
For the accuracy results that follow, it will be easiest to work directly with PDE \eqref{eq:hjb.5a} for $\eta$ and PDE \eqref{eq:hjb.5} for $\psi$, which are related to $u$ via \eqref{eq:u=eta.phi}, rather than work with the PDE \eqref{eq:hjb.5b} for $u$.  As we have not yet performed an asymptotic analysis of PDE \eqref{eq:hjb.5} for $\psi$, we introduce
\begin{align}
0
	&=	\( \d_t + \Act^\eps \) \psi^\eps + \Bc^\eps (\psi^\eps) , &
\psi^\eps(T,x,y,\nu)
	&=	-\gam \nu \varphi(x) , &
\eps
	&\in [0,1] . \label{e1}
\end{align}
Following the asymptotic analysis in Section \ref{sec:psi.asymptotics}, we seek a solution $\psi^\eps$ of the form
\begin{align}
\psi^\eps
	&=	\sum_{n=0}^\infty \eps^n \psi_n . \label{e2}
\end{align}
We insert expansion \eqref{e2} for $\psi^\eps$ into PDE \eqref{e1} and collect terms of like powers in $\eps$.  At lowest order we obtain
\begin{align}
\Oc(1):&&
0
	&=	(\d_t + \Act_0) \psi_0 + (1- \rho^2) ( \tfrac{1}{2} \beta^2)_0 (\d_y \psi_0)^2 - (\tfrac{1}{2}\lam^2)_0 , &
\psi_0(T,x,y,\nu)
	&=	-\gam \nu \varphi(x)  , \label{e3} 
\end{align}
Suppose that $\psi_0$ is a function of $(t,x,\nu)$ only.  Then \eqref{e3} becomes
\begin{align}
\Oc(1):&&
0
	&=	(\d_t + \Act_0) \psi_0  - (\tfrac{1}{2}\lam^2)_0 , &
\psi_0(T,x,\nu)
	&=	-\gam \nu \varphi(x)  . \label{e4}
\end{align}
One can easily verify that the solution $\psi_0$ of \eqref{e4} is
\begin{align}
\psi_0
	&=	\eta_0 - \gam \nu u_0 . \label{e5}
\end{align}
Thus, we identify the expression in \eqref{e5} as the solution $\psi_0$ of \eqref{e3}.
Now, we subtract \eqref{e4} from \eqref{eq:hjb.5} and obtain
\begin{align}
0
	&=	(\d_t + \Act) ( \psi - \psi_0 ) -  (\Act_0 - \Act) \psi_0 \\ & \quad
			+ (1-\rho^2)(\tfrac{1}{2} \beta^2) (\d_y \psi)^2 - \( (\tfrac{1}{2} \lam^2) - (\tfrac{1}{2} \lam^2)_0 \) , &
0
	&=	( \psi - \psi_0 )(T,x,y,\nu) . \label{e6}
\end{align}
By Duhamel's principle, we can express $(\psi - \psi_0)$ as follows
\begin{align}
( \psi - \psi_0 )(t,x,y)
	&=	- \int_t^T \dd s \int_{\Rb^2} \dd \xi \dd \eta \, \Gam(t,x,y;s,\xi,\eta) (\Act_0 - \Act) \psi_0(s,\xi,\eta) 
				\Big\} =: A \\ & \quad
			+ (1-\rho^2) \int_t^T \dd s \int_{\Rb^2} \dd \xi \dd \eta \, \Gam(t,x,y;s,\xi,\eta) (\tfrac{1}{2} \beta^2)(s,\xi,\eta) (\d_y \psi(s,\xi,\eta))^2 
				\Big\} =: B \\ & \quad
			- \int_t^T \dd s \int_{\Rb^2} \dd \xi \dd \eta \, \Gam(t,x,y;s,\xi,\eta) \( (\tfrac{1}{2} \lam^2)(\xi,\eta) - (\tfrac{1}{2} \lam^2)_0 \) .
				\Big\} =: C 
\end{align}
where $\Gam$ is the fundamental solution corresponding to parabolic operator $(\d_t + \Act)$.  We wish to bound $A$, $B$, and $C$.  To do this, we must impose some conditions on the coefficients of $\Act$ and the terminal data $\varphi$.
\begin{assumption}
\label{assumption.X}
Denote by $C_b^{m,1}$ the class of bounded functions whose derivatives up to order $m$ are Lipschitz continuous.  Define the following norms
\begin{align}
\norm{f}_{C_b^{m,1}}
	&=	\sum_{k=0}^m \norm{\d^k f}_\infty , &
	&\text{with}&
\norm{f}_{C_b^{-1,1}}
	&=	\norm{f}_\infty .
\end{align}
We assume there exists a constant $M>0$ such that the following holds: \\
(i) \emph{Uniform ellipticity}: 
\begin{align}
(1/M) |\xi|^2 
	&\leq \sum_{i,j=1}^2 \mathbf{A}_{i,j}(x,y) \xi_i \xi_j \leq M |\xi^2|, &
	&\forall \xi, (x,y) \in \Rb^2 , &
\mathbf{A}
	&=	\( \begin{array}{cc}
			(\tfrac{1}{2}\sig^2) & (\rho \sig \beta) \\
			(\rho \sig \beta) & (\tfrac{1}{2} \beta^2) 
			\end{array} \) .
\end{align}
(ii) \emph{Regularity and boundedness}: The coefficients $(\tfrac{1}{2}\lam^2)$, $(\tfrac{1}{2}\beta^2)$, $(\tfrac{1}{2}\sig^2)$, $(\rho \sig \beta)$ and $(c - \rho \beta \lam)$ are $C_b^{m,1}(\Rb)$ with their norms $\norm{\cdot}_{C_b^{m,1}}$ bounded by $M$.\\
(iii) The terminal data satisfies $\varphi \in C_b^{k-1,1}$ for some $0 \leq k \leq 2$.
\end{assumption}

In the rest of Section \ref{sec:accuracy.X} (and Section \ref{sec:accuracy.X} \emph{only}), Assumption \ref{assumption.X} is in force.
Let us introduce $\tau = T-t$.
Using classical properties of $\Gam$, \cite[Theorem 3.10]{lorig-pagliarani-pascucci-4} establishes that $A = \Oc(\tau^{\tfrac{1+k}{2}})$ and $C = \Oc(\tau^{\tfrac{3}{2}})$ as $\tau \to 0$.  In order to bound $B$, we note that $(\d_y \psi)$ is uniformly bounded in $[0,T] \times \Rb^2$ by \cite[Theorem V.8.1]{ladyzhenskaia1988linear}.  And thus, by Assumption \ref{assumption.X} (ii), we have that $C = \Oc(\tau)$ as $\tau \to 0$.  Since $( \psi - \psi_0 ) = A+B+C$, it follows that 
$( \psi - \psi_0 ) = \Oc(\tau) \vee \Oc(\tau^{\tfrac{1+k}{2}})$.
A similar analysis for $\eta$ reveals that $( \eta - \eta_0 ) = \Oc(\tau)$.  Thus, we have
\begin{align}
u - \ub_0
	&=	u - u_0 
	=		\frac{1}{\gam \nu} \Big( (\eta - \psi) - (\eta_0 - \psi_0) \Big)
	=		\frac{1}{\gam \nu} \Big( (\eta - \eta_0) - (\psi - \psi_0) \Big) \\
	&=		\frac{1}{\gam \nu} \(  \Oc(\tau) \vee \Oc(\tau^{\tfrac{1+k}{2}}) - \Oc(\tau) \)
	=		\Oc(\tau) \vee \Oc(\tau^{\tfrac{1+k}{2}})  .
\end{align}
Hence, we have established that the zeroth order indifference price approximation $\ub_0$ satisfies
\begin{align}
|u - \ub_0 |
	&=	\Oc(\tau) \vee \Oc(\tau^{\tfrac{1+k}{2}}) &
	&		\text{as $\tau \to 0$} . \label{eq:accuracy.u0}
\end{align}
The steps needed to establish the accuracy of the $m$th order approximations $\ub_m$ are similar.
First, find the PDEs satisfied by $(\psi - \psib_m)$ and $(\eta - \etab_m)$.  Next, use Duhamel's principle to express $(\psi - \psib_m)$ and $(\eta - \etab_m)$ as integrals involving $\Gam$.  Then, use classical properties of $\Gam$ to bound the result.  Finally, use 
\begin{align}
u - \ub_m
	=		\frac{1}{\gam \nu} \Big( (\eta - \etab_m) - (\psi - \psib_m) \Big) , 
\end{align}
to bound $(u - \ub_m)$.  The details of the accuracy proof for higher order approximations are obviously much more involved due to the increased number of terms that must be analyzed.

%
%

\subsection{Implied volatility asymptotics}
\label{sec:impvol}
In this section, we translate our expansion for the indifference price of a call into an expansion for implied volatility.  
\begin{assumption}
\label{ass:imp.vol}
Throughout this section we fix an LSV model $(X,Y)$, a time $t$, a maturity date $T>t$, initial values $(X_t,Y_t)=(x,y)$, a call payoff $\varphi(X_T)=(\ee^{X_T} - \ee^k)^+$, and a number of contracts $\nu$.
\end{assumption}
\noindent
Assumption \ref{ass:imp.vol} fixes an indifference price $u$ as the solution of \eqref{eq:hjb.5b}.  Our goal is to find the implied volatility corresponding to \emph{this particular indifference call price}.  To ease notation, in what follows, we will suppress much of the dependence on $(t,T,x,y,k,\nu)$. However, the reader should keep in mind that the implied volatility of the option under consideration \emph{does} depend on $(t,T,x,y,k,\nu)$, even if this is not explicitly indicated.  To begin, we provide a definition for implied volatility, which will be fundamental throughout this section.
\begin{definition}
\label{def:imp.vol.def} 
For a fixed $(t,x,T,k)$, the \emph{implied volatility} corresponding to call price $p \geq (\ee^x - \ee^k)$ is defined as the unique strictly positive real solution $\II$ of the equation
\begin{align}
u^\BS(\II)    
	&=	p.   \label{eq:imp.vol.def}
\end{align}
where $u^\BS$, defined in \eqref{eq:uBS}, is regarded as a function of volatility only. 
\end{definition}
\noindent
Our goal is find the implied volatility $\II$ corresponding to the indifference price $u$ of a call option.  To do this, we introduce \emph{modified implied volatility} $\II^\eps$, the solution to 
\begin{align}
u^\BS(\II^\eps)    
	&=	u^\eps , &
\eps
	&\in [0,1] , \label{eq:imp.vol.eps}
\end{align}
where $u^\eps$ solves \eqref{eq:eta.eps.pde}.  We will seek an asymptotic solution $\II^\eps$ to \eqref{eq:imp.vol.eps} by expanding $\II^\eps$ in powers of $\eps$:
\begin{align}
\II^\eps
	&=	\sum_{i=0}^\infty \eps^i \II_i . \label{eq:sig.expand}
\end{align}
Our asymptotic solution to $u^\BS(\II)=u$,
which is the case we are interested in, will be obtained by setting $\eps=1$.  
We insert expressions \eqref{eq:u.expand} and \eqref{eq:sig.expand} into \eqref{eq:imp.vol.eps}, expand both sides in powers of $\eps$ and collect terms of like order.  We obtain
\begin{align} 
\Oc(1):&&
u_0
	&=	u^\BS(\II_0) , \\
\Oc(\eps^k):&&
u_m
	&=		\II_m \d_\sig u^{\BS}(\II_0)
				+ \sum_{k=2}^m \frac{1}{k!}\Big( \sum_{I_{m,k}} \prod_{j=1}^k \II_{i_j} \Big)  \d_\sig^k	u^{\BS}(\II_0) , &
m 
	&\geq 1 , \\ & &
I_{m,k}
	&:=	\Big\{ i = (i_1,i_2,\cdots,i_k) \in \mathbb{N}^k : \sum_{j=1}^k i_j = m \Big\} . \label{eq:Imk}
\end{align}
Noting from Theorem \ref{thm:u} that $u_0 = u^\BS(\sig_0)$ we can solve the above equations recursively
\begin{align}
\Oc(1):&&
\II_0
	&=	\sig_0 , \label{eq:I0} \\
\Oc(\eps^m):&&
\II_m
	&=	\frac{1}{\d_\sig u^{\BS}(\II_0)} 
			\bigg( u_m - \sum_{k=2}^m \frac{1}{k!}\Big( \sum_{I_{m,k}} \prod_{j=1}^k \II_{i_j} \Big)  \d_\sig^k	u^{\BS}(\II_0)  \bigg) , &
m
	&\geq 1 . \label{eq:Im}
\end{align}
Explicitly, the order $\Oc(\eps)$ and $\Oc(\eps^2)$ terms are given by
\begin{align}
\Oc(\eps):&&
\II_1
	&=	\frac{1}{\d_\sig u^\BS(\II_0)} u_1,  &
\Oc(\eps^2):&&
\II_2
	&= 	\frac{1}{\d_\sig u^\BS(\II_0)} \Big( u_2 - \tfrac{1}{2} \II_1^2 \d_\sig^2 u^\BS(\II_0) \Big) ,  \label{eq:I2}
\end{align}
Having served its purpose, we now set $\eps=1$, and make the following definition
\begin{definition}
The \emph{$m$th-order approximation of the implied volatility} $\IIb_m$ corresponding to the indifference price $u$ of a call option is defined as
\begin{align}
\IIb_m
	&=	\sum_{i=0}^m \II_i , \label{eq:Ibar.2}
\end{align}
where $\II_0$ and $\II_i$ $(i \geq 1)$ are given by \eqref{eq:I0} and \eqref{eq:Im}, respectively.
\end{definition}  
\begin{remark}
\label{rmk:iv.exact}
In general, our expansion for the indifference price $u$ is asymptotic -- not convergent.  And thus, our implied volatility expansion is asymptotic as well.  However, it is worth noting that if one has a convergent expansion for a call price $u^\eps$ of the form 
\begin{align}
u^\eps
	&=	u^{\BS}(\sig_0) + \sum_{n=0}^\infty \eps^n u_n ,
\end{align}
then the resulting implied volatility expansion $\II^\eps = \sum_{n=0}^\infty \eps^n \II_n$ is exact so long as $u^\eps$ falls within the radius of convergence of the power series expansion of $[u^{\BS}]^{-1}$ about $u^{\BS}(\sig_0)$.  For the interested reader, convergence is proved in \cite[Theorem 4.3 and Remark 9]{lorig-3}.
\end{remark}
\noindent
Note that expression for $\II_i$ $(i \geq 1)$ contains $u_j$ $(j \leq i-1)$.  Since $u_i$ and $(i \geq 1)$ are computed as differential operators acting on $u_0=u^\BS(\sig_0)$, it is difficult to discern how the implied volatility surface behaves as a function of time to maturity $T-t$ and $\log$ strike $k$.  The following proposition provides an explicit expression for $\II_1$ and $\II_2$ in terms of the $\log$-strike $k$ and time to maturity $\tau:=T-t$.
\begin{proposition}
\label{thm:imp.vol}
Let Assumption \ref{ass:imp.vol} hold.  Denote time to maturity by $\tau := T-t$ and $\log$ moneyness by $L = k-x$.  Fix the expansion point of the Taylor series expansion $(\xb,\yb)=(x,y)$.  Define
\begin{align}
a
	&=	\tfrac{1}{2} \sig^2 , &
b
	&=	\tfrac{1}{2} \beta^2 , &
f
	&=	c - \rho \beta \lam , &
g
	&=	\rho \sig \beta . \label{eq:abfg}
\end{align}
Assume $(\tfrac{1}{2}\lam^2)$ is $C^1(\Rb)$ and $a$, $b$, $f$ and $g$ are $C^2(\Rb^2)$.
Then, $\II_1$ and $\II_2$, defined in \eqref{eq:I2}, are given by
\begin{align}
\II_1
	&=	\II_{1,0} + \II_{0,1} , &
\II_2
	&=	\II_{2,0} + \II_{1,1} + \II_{0,2} + \II_2^\Ind ,
\end{align}
where
\begin{align}
\II_{1,0}
    &=  \(  \frac{a_{1,0}}{2 \sigma _0} \) L, &
\II_{0,1}
    &=  \tau \( \frac{a_{0,1} \left(g_{0,0}+2 f _{0,0}\right)}{4 \sigma _0} \)
            + \( \frac{a_{0,1} g_{0,0}}{2 \sigma _0^3} \) L ,
\end{align}
and
\begin{align}
\II_{2,0}
    &=  \tau \Big( \frac{1}{12} \sigma _0 a_{2,0}-\frac{a_{1,0}^2}{8 \sigma _0} \Big)
				+ \tau^2 \Big( -\frac{1}{96} \sigma _0 a_{1,0}^2 \Big)
				+ \Big( \frac{2 \sigma _0^2 a_{2,0}-3 a_{1,0}^2}{12 \sigma _0^3} \Big) L^2 , \\
\II_{1,1}
    &=  \frac{\tau}{12 \sigma _0^3} \Big( \sigma _0^2 a_{1,1} g_{0,0}+a_{0,1} \left(a_{1,0} g_{0,0}-2 \sigma _0^2 g_{1,0}\right) \Big) 
				+ \frac{\tau^2}{48 \sigma _0} \Big( -a_{0,1} a_{1,0} g_{0,0} \Big) \\ &\quad
				+ \frac{\tau}{24 \sigma _0^3} \Big( 2 \sigma _0^2 a_{1,1} \left(g_{0,0}+2 f_{0,0}\right)+a_{0,1} \left(2 \sigma _0^2 \left(g_{1,0}+2 f_{1,0}\right)-5 a_{1,0} \left(g_{0,0}+2 f_{0,0}\right)\right) \Big) L \\ &\quad
				+ \frac{1}{6 \sigma _0^5} \Big( \sigma _0^2 a_{1,1} g_{0,0}+a_{0,1} \left(\sigma _0^2 g_{1,0}-5 a_{1,0} g_{0,0}\right) \Big) L^2, \\
\II_{0,2}
    &=   \frac{\tau}{24 \sigma _0^5} \Big( 4 \sigma _0^2 a_{0,2} \left(3 \sigma _0^2 b_{0,0}-g_{0,0}^2\right)+a_{0,1} \left(a_{0,1} \left(9 g_{0,0}^2-8 \sigma _0^2 b_{0,0}\right)-4 \sigma _0^2 g_{0,0} g_{0,1}\right) \Big) \\ &\quad
				+ \frac{\tau^2}{24 \sigma _0^3} \Big( a_{0,1} \left(-2 \sigma _0^2 a_{0,1} b_{0,0}+g_{0,0} \left(\sigma _0^2 \left(g_{0,1}+2 f_{0,1}\right)-3 a_{0,1} f_{0,0}\right)\right)
					\\ & \quad \quad
				+	a_{0,1} f_{0,0} \left(2 \sigma _0^2 \left(g_{0,1}+2 f_{0,1}\right)-3 a_{0,1} f_{0,0}\right)
				+\sigma _0^2 a_{0,2} \left(g_{0,0}+2 f_{0,0}\right){}^2 \Big) \\ &\quad
				+ \frac{\tau}{24 \sigma _0^5} \Big( a_{0,1} \left(g_{0,0} \left(4 \sigma _0^2 \left(g_{0,1}+f_{0,1}\right)-18 a_{0,1} f_{0,0}\right)-9 a_{0,1} g_{0,0}^2+4 \sigma _0^2 g_{0,1} f_{0,0}\right) 
				\\ & \quad \quad
				+ 4 \sigma _0^2 a_{0,2} g_{0,0} \left(g_{0,0}+2 f_{0,0}\right)
				\Big) L \\ &\quad
				+ \frac{1}{12 \sigma _0^7} \Big( a_{0,1} \left(a_{0,1} \left(4 \sigma _0^2 b_{0,0}-9 g_{0,0}^2\right)+2 \sigma _0^2 g_{0,0} g_{0,1}\right)+2 \sigma _0^2 a_{0,2} g_{0,0}^2 \Big)L^2 , \\
\II_{2}^\Ind
	&=	(1 - \rho^2) (\tfrac{1}{2} \beta^2)_0 \bigg\{  \frac{- 2 (\tfrac{1}{2}\lam^2)_{0,1} ( \tfrac{1}{2} \sig^2)_{0,1} \tau^2 }{3 \sig_0}
			\\ & \quad
			- \gam \nu  \frac{(\tfrac{1}{2}\sig^2)_{0,1}^2 }{\sig_0^2\sqrt{2\pi}}
			\frac{\ee^k}{\sqrt{\tau}} \exp \( \frac{\( L + \tfrac{1}{2}\sig_0^2 \tau \)^2 }{2 \sig_0^2\tau } \)
			\int_t^T \dd t_1  \frac{(T-t_1)^{3/2}}{\sqrt{\tau+t_1-t}}
			\exp \(\frac{-\(L+\frac{1}{2} \sig_0^2 \tau\)^2}{\sig_0^2 (\tau+t_1-t)} \) \bigg\} . \label{eq:I2.indiff}
\end{align}
\end{proposition}
\begin{proof}
The proof is given in Appendix \ref{sec:imp.vol.proof}.
\end{proof}
\begin{remark}
Observe that we have isolated the term $\II_2^\Ind=u_2^\Ind/\d_\sig u^\BS$.  This term results from the nonlinear aspect of indifference pricing and gives rise to a bid-ask spread in the implied volatility surface.  Specifically, the approximate bid-ask spread is given by twice the absolute value of the second term in \eqref{eq:I2.indiff}.
\end{remark}
\begin{remark}
As $\tau \to 0$ we have $\II_2^\Ind \to 0$ because for $t_1 \in (t,T)$ we have
\begin{align}
\exp\( \frac{1}{2 \tau} - \frac{1}{\tau+t_1-t} \) 
	&\leq \exp\( \frac{1}{2 \tau} - \frac{1}{2 \tau} \) = 1 , &
	&\text{and}&
\frac{(T-t_1)^{3/2}}{\sqrt{\tau}\sqrt{\tau+t_1-t}} &\leq \frac{\tau^{3/2}}{\tau} = \sqrt{\tau} .
\end{align}
\end{remark}
\begin{remark}
As we found with $u_2^\Ind$, the second term in $\II_2^\Ind$ has the sign of $(-\nu)$.  Thus, the approximation of the implied volatility surface of a European call buyer will be strictly \emph{below} the approximation of the implied volatility surface of a European call writer.  
\end{remark}

\section{Options on non-traded assets}
\label{sec:Y}
In Section \ref{sec:asymptotics} we derived an expression for the approximate indifference price of a European-style claim whose payoff function $\varphi$ depended only on the terminal value of a traded asset $X_T$.  We now consider a European-style claim whose payoff function $\varphi$ depends only on the terminal value of a non-traded underlying $Y_T$.
Examples of options on non-traded underlyings include (i) employee stock options (see \cite{henderson2005impact,leung2009accounting}), where an executive is awarded options on his employer's stock, but is limited from trading the stock due to regulatory rules (ii) options on the price of a physical commodity, such a gold, silver or soybeans (\cite{geman2009commodities}), and (iii) weather derivatives, which have payoffs that depend on realized weather statistics, such as average or minimum temperature and accumulated rainfall (\cite{alaton2002modelling}).
\begin{assumption}
\label{ass:y}
Throughout Section \ref{sec:Y}, we consider dynamics of the form
\begin{align}
\begin{aligned}
\dd X_t
	&=	\Big( \mu(Y_t) - \tfrac{1}{2}\sig^2(Y_t) \Big)  \dd t + \sig(Y_t)  \dd B_t^X , \\
\dd Y_t
	&=	c(Y_t) \dd t + \beta(Y_t) \Big( \rho  \dd B_t^X + \sqrt{1- \rho^2}  \dd B_t^Y \Big) . 
\end{aligned} \label{eq:SV} 
\end{align}
We further assume the payoff function $\varphi$ of the European claim is a function of $y$ only.
\end{assumption}
\noindent

\subsection{PDE Asymptotics}
Our goal is to find the indifference price $u$ of a European-style contingent claim with payoff $\varphi(Y_T)$.
As there are no explicit formulas for the indifference price 
under general dynamics of the form \eqref{eq:SV}, we shall seek an approximation for $u$.
Since $Y$ is not traded, it is not constrained to be a martingale under any risk-neutral measure.  Thus, it makes no sense to derive implied volatility asymptotics for the claim with payoff $\varphi(Y_T)$.
\par
In Section \ref{sec:asymptotics} we obtained an approximation for the indifference price $u$ in two steps: first, we obtained an approximation for $\eta$, the solution of PDE \eqref{eq:hjb.5a}.  Then, we found an approximation for $u$, the solution of PDE \eqref{eq:hjb.5b}.  In this section, we follow an alternative approach.  Here, we shall seek an approximation for $\psi$, the solution of PDE \eqref{eq:hjb.5}.  Then, using $\eta = \psi |_{\nu=0}$ we will obtain an approximation for $u$ via \eqref{eq:u=eta.phi}.
\par
As in Section \ref{sec:asymptotics}, to find an approximation for $\psi$, we introduce $\psi^\eps$, the solution to
\begin{align}
0
	&=	\( \d_t + \Act^\eps \) \psi^\eps + \Bc^\eps (\psi^\eps) , &
\psi^\eps(T,y,\nu)
	&=	-\gam \nu \varphi(y), &
\eps
	&\in [0,1] , \label{eq:PDE.zeta}
\end{align}
where $\psi^\eps$ depends only on $(t,y,n)$ only since neither the coefficients $(\mu,\sig,c,\beta)$ nor in the payoff function $\varphi$ depend on $x$.  As always, we will seek an asymptotic solution to \eqref{eq:PDE.zeta} by expanding in powers of $\eps$.  The approximate solution to PDE \eqref{eq:hjb.5}, which is the case we are interested in, will be obtained by setting $\eps=1$.
\par
Comparing PDE \eqref{eq:PDE.zeta} with \eqref{eq:psi.eps.pde}, the similarity of these two PDEs might suggest that one can obtain an approximation solution to \eqref{eq:PDE.zeta} with only minor modifications to the approximation we previously obtained for the solution to PDE \eqref{eq:psi.eps.pde}. 
However, this is not the case.  The difficulty in finding an approximation solution to \eqref{eq:PDE.zeta} arises from the terminal condition $\psi^\eps(T,y,\nu) =	-\gam \nu \varphi(y)$.  Because the terminal condition depends on $y$, the order zero approximation $\psi_0$ must also depend on $y$.  
This was not the case for in PDE \eqref{eq:psi.eps.pde}, where the terminal condition $\eta^\eps(T,x,y)=0$ led to a zeroth order approximation $\eta_0(t)$ independent of $y$.  Since the zeroth order solution to \eqref{eq:PDE.zeta} must depend on $y$, the nonlinear term $(\d_y \psi^\eps)^2$, which appears in $\Bc^\eps(\psi^\eps)$, remains in the zeroth order PDE.  However, the nonlinearity in PDE \eqref{eq:PDE.zeta} can be removed.  Following \cite{zariphopoulou2001} we set
\begin{align}
\psi^\eps
	&=	\frac{1}{(1-\rho^2)} \log \xi^\eps , \label{eq:psi.xi}
\end{align}
Inserting \eqref{eq:psi.xi} into \eqref{eq:PDE.zeta} we see that $\xi^\eps$ satisfies
\begin{align}
0
	&=	(\d_t + (\Act')^\eps ) \xi^\eps , &
\xi^\eps(T,y,\nu)
	&=	\theta(y,\nu) , \label{eq:xi.eps.pde} \\ 
(\Act')^\eps
	&:=	\Act^\eps - (1-\rho^2)(\tfrac{1}{2}\lam^2)^\eps, &	
\theta(y,\nu)
	&:=	\exp \( -\gam \nu (1-\rho^2) \varphi(y) \) .
\end{align}
Note that \eqref{eq:xi.eps.pde} is a linear parabolic Cauchy problem.  We can obtain an approximation for $\xi^\eps$ using the Taylor series expansion methods described in Section \ref{sec:asymptotics}.  
Noting that $(\Act')^\eps$ can be written as a power series in $\eps$, we suppose the solution $\xi^\eps$ can also be written in this form
\begin{align}
(\Act')^\eps
	&=	\sum_{i=0}^\infty \eps^i \Act_i' , &
\Act_i'
	&=	\Act_i - (1-\rho^2)(\tfrac{1}{2}\lam^2)_i , &
\xi^\eps
	&=	\sum_{i=0}^\infty \eps^i \xi_i  \label{eq:xi.expand}.
\end{align}
Inserting \eqref{eq:xi.expand} into \eqref{eq:xi.eps.pde} and collecting terms of like powers of $\eps$, we find
\begin{align}
\Oc(1):&&
0
	&= (\d_t + \Act_0')\xi_0 , &
\xi_0(T,y,\nu)
	&=	\theta(y,\nu) , 	\label{eq:xi0.pde} \\
\Oc(\eps^m):&&
0
	&= (\d_t + \Act_0')\xi_m + \sum_{k=1}^m \Act_k' \xi_{m-k} , &
\xi_m(T,y,\nu)
	&=	0  . \label{eq:xim.pde} 
\end{align}
The above sequence of nested Cauchy problems has been solved explicitly in \cite{lorig-pagliarani-pascucci-2}.  We briefly review the results in the next Section.

\subsection{Explicit expressions}
Let $\Pc_0'(t,T)$ be the semigroup of operators generated by $\Act_0'$ and construct $\Gc_i'(t,T)$ from $\Act_i'$ in the same manner that $\Gc_i(t,T)$ in \eqref{eq:Gc} is constructed from $\Act_i$.  Specifically, we define
\begin{align}
\Pc_0'(t,T)
	&:=	\exp \Big( -(T-t) (1-\rho^2)(\tfrac{1}{2}\lam^2)_0 \Big) \Pc_0(t,T) , &
\Gc_i'(t,T)
	&:=	\Act_i'(\Yc(t,T)) .
\end{align}
We are now in a position to provide and explicit expression for $\xi_i$ $(i \geq 0)$.
\begin{proposition}
\label{thm:xi}
Assume the coefficients $(\tfrac{1}{2}\lam^2)$, $(c - \rho \beta \lam)$ and $(\tfrac{1}{2}\beta^2)$ are $C^m(\Rb)$ and $\theta(\cdot,\nu)$ is at most exponentially growing.  Then the unique classical solutions of \eqref{eq:xi0.pde} and \eqref{eq:xim.pde} are given by (omitting the arguments $(y,\nu)$ for clarity)
\begin{align}
\xi_0(t)
	&=	\Pc_0'(t,T) \theta , &
\xi_m(t)
	&=	\Lc_m'(t,T) \xi_0(t) ,
\end{align}
where the operator $\Lc_0'(t,T)$ is given by
\begin{align}
\Lc_m'(t,T)
	&:=	\sum_{k=1}^m \int_t^T \dd t_1 \int_{t_1}^T \dd t_2 \cdots \int_{t_{k-1}}^T \dd t_k \sum_{I_{m,k}}
			\Gc_{i_1}'(t,t_1) \Gc_{i_2}'(t,t_2) \cdots \Gc_{i_k}'(t,t_k) , 
\end{align}
where $I_{m,k}$ is defined in \eqref{eq:Imk}.
\end{proposition}
\begin{proof}
See \cite[Theorem 7]{lorig-pagliarani-pascucci-2}.
\end{proof}
\noindent
We now wish to translate our expansion of $\xi^\eps$ into an expansion for $\psi^\eps$.  Expanding \eqref{eq:psi.xi} in powers of $\eps$ we obtain 
\begin{align}
\Oc(1):&&
\psi_0
	&:=	\frac{1}{(1-\rho^2)} \log \xi_0 , \label{eq:psi0.again} \\
\Oc(\eps^m):&&
\psi_m
	&:=\frac{1}{(1-\rho^2)}  \sum_{k=1}^m \frac{(-1)^{k-1}}{k} \xi_0^{-k} 
			\( \sum_{i \in I_{m,k}} \prod_{j=1}^k \xi_{i_j} \) , \label{eq:psik.again} 
\end{align}
where $I_{m,k}$ is defined in \eqref{eq:Imk}.  Note that $\eta_i(t,y) = \psi_i(t,y,0)$ for every $i \geq 0$.  With \eqref{eq:u=eta.phi} in mind, we now define our $m$th-order approximation of the indifference price.
\begin{definition}
Let Assumption \ref{ass:y} hold.  
Assume the coefficients $(\tfrac{1}{2}\lam^2)$, $(c - \rho \beta \lam)$ and $(\tfrac{1}{2}\beta^2)$ are $C^m(\Rb)$ and $\theta(\cdot,\nu)$ is at most exponentially growing.
Then the \emph{$m$th-order approximation of the indifference price} is given by
\begin{align}
\ub_m
	&=	\frac{1}{\gam \nu} \big( \etab_m - \psib_m \big) , &
\psib_m
	&:=	\sum_{i=0}^m \psi_i , &
\etab_m(t,y)
	&:=	\psib_m(t,y,0) , 	\label{eq:ubar.m}
\end{align}
where $\psi_0$ is given by \eqref{eq:psi0.again} and $\psi_i$ $(i \geq 1)$ is given by \eqref{eq:psik.again}.
\end{definition}


\subsection{Accuracy of the indifference price approximation}
In this section, we will establish error estimates for $\ub_m$, the $m$th-order approximation of the indifference price $u$.
To establish these estimates, we require the following assumptions:
\begin{assumption}
\label{assumption}
We assume there exists a constant $M>0$ such that the following holds: \\
(i) \emph{Uniform ellipticity}: $1/M \leq (\tfrac{1}{2} \beta^2) \leq M$.\\
(ii) \emph{Regularity and boundedness}: The coefficients $(\tfrac{1}{2}\lam^2)$, $(\tfrac{1}{2}\beta^2)$ and $(c - \rho \beta \lam)$ are $C_b^{m,1}(\Rb)$ with their norms $\norm{\cdot}_{C_b^{m,1}}$ bounded by $M$,
where $C_b^{m,1}(\Rb)$ and $\norm{\cdot}_{C_b^{m,1}}$ are defined in Assumption \ref{assumption.X}.\\
(iii) The terminal data $\theta$ satisfies $\theta(\cdot,\nu) \in C_b^{k-1,1}$ for some $0 \leq k \leq 2$.
\end{assumption}
\begin{theorem}
\label{thm:xi.bound}
Let Assumption \ref{assumption} hold.  Let $\xi$ be the solution of Cauchy problem \eqref{eq:xi.eps.pde} with $\eps=1$.  
Define
\begin{align}
\xib_m
	&:=	\sum_{i=0}^m \xi_i , &
\tau
	&:=	T-t ,
\end{align}
Then
\begin{align}
\xi_0
	&=	\Oc(1)  &
	&\text{and}&
\sup_y |\xi(t,y,\nu) - \xib_m(t,y,\nu)|
	&=	\Oc ( \tau^{\frac{m+k+1}{2}} )  &
	&\text{as $\tau \to 0$},  \label{eq:xi.accuracy}
\end{align}
where $\xi_i$ are as given Proposition \eqref{thm:xi} and $k$ is the constant that appears in Assumption \ref{assumption} (iii).
\end{theorem}
\begin{proof}
See \cite[Theorem 3.10 and Proposition 6.22]{lorig-pagliarani-pascucci-4}.
\end{proof}
\noindent
Using Theorem \ref{thm:xi.bound} we can establish the accuracy of $\ub_m$, the $m$th-order approximation of the indifference price $u$.
\begin{corollary}
\label{cor:u.bound}
Let Assumptions \ref{ass:y} and \ref{assumption} hold.  Then $\ub_m$, 
given by \eqref{eq:ubar.m}, satisfies
\begin{align}
\sup_y | u(t,y,\nu) - \ub_m(t,y,\nu) | 
	&=	\Oc(\tau^{\frac{m+k+1}{2}})  &
	&\text{as $\tau \to 0$}. \label{eq:u.accuracy}
\end{align}
\end{corollary}
\begin{proof}
First, we note that equation \eqref{eq:xi.accuracy} implies
\begin{align}
\sup_y \xi_m(t,y,\nu)
	&=	\Oc( \tau^{\frac{m+k}{2}} ) , &
m
	&\geq 1 . \label{eq:xi.m.accuracy}
\end{align}
Next, we have
\begin{align}
\psi^\eps
	&=	\frac{1}{1-\rho^2} \log \xi^\eps &
	&		\text{(by \eqref{eq:psi.xi})} \\
	&=	\frac{1}{1-\rho^2} \log \( \sum_{i=0}^m \eps^n \xi_i + \Oc(\tau^{\frac{m+1+k}{2}}) \) &
	&		\text{(by \eqref{eq:xi.m.accuracy})} \\
	&=	\frac{1}{1-\rho^2} \xi_0 
			+ \sum_{n=1}^m \frac{\eps^n}{(1-\rho^2)}  \sum_{k=1}^n \frac{(-1)^{k-1}}{k} \xi_0^{-k} 
			\( \sum_{i \in I_{n,k}} \prod_{j=1}^k \xi_{i_j} \)
			+ \Oc(\tau^{\frac{m+1+k}{2}}) &
	&		\text{(expanding in powers of $\eps$)} \\
	&=	\sum_{n=0}^m \eps^n \psi_n + \Oc(\tau^{\frac{m+1+k}{2}}) . &
	&		\text{(by \eqref{eq:psi0.again} and \eqref{eq:psik.again})} \label{eq:new.result}
\end{align}
Therefore, we find
\begin{align}
\psi - \psib_m
	&=	\psi - \sum_{n=0}^m \psi_n &
	&		\text{(by \eqref{eq:ubar.m})} \\
	&=	\Oc(\tau^{\frac{m+1+k}{2}}) . &
	&		\text{(setting $\eps = 1$ in \eqref{eq:new.result})} \label{eq:newer.result}
\end{align}
Taking the absolute value of both sides and then taking a $\sup_y$, it follows from \eqref{eq:newer.result} that
\begin{align}
\sup_y |\psi(t,y,\nu)-\psib_m(t,y,\nu)|
	&=	\Oc(\tau^{\frac{m+k+1}{2}})  . \label{eq:psi.accuracy}
\end{align}
Finally, using \eqref{eq:u=eta.phi} and \eqref{eq:ubar.m},  as well as $\eta(t,y) = \psi(t,y,0)$ and $\etab_m(t,y) = \psib_m(t,y,0)$ we have (omitting the arguments $(t,y,\nu)$ for clarity)
\begin{align}
\sup_y | u - \ub_m |
	&=	\sup_y \Big| \frac{1}{\gam \nu} ( \eta - \psi ) - \frac{1}{\gam \nu}( \etab_m - \psib_m ) \Big|
	\leq	\frac{1}{|\gam \nu|} \Big( \sup_y |\eta - \etab_m| + \sup_y |\psi - \psib_m| \Big) ,
\end{align}
which, combined with \eqref{eq:psi.accuracy}, yields the claimed accuracy result \eqref{eq:u.accuracy}.
\end{proof}

\begin{remark}
\label{rmk:accuracy}
We wish to clear up a common point of confusion.
Theorem \ref{thm:xi.bound} and Corollary \ref{cor:u.bound} are asymptotic accuracy results, which provide information about how quickly $\xib_m$ and $\ub_m$ approach $\xi$ and $u$, respectively, as $\tau \to 0$.  In particular, a larger $m$ implies a faster rate of convergence since both
$|\psi-\psib_m|$ and $|u-\ub_m|$ 
are of order $\Oc(\tau^{\frac{m+k+1}{2}})$ as $\tau \to 0$.  Small-time results of the form given in Theorem \ref{thm:xi.bound} and Corollary \ref{cor:u.bound} are common in the asymptotic expansion literature; see, for example, 
\cite[Theorem 4.3]{henry2008analysis},
\cite[equation (63)]{hagan2005probability},
or the various accuracy results in \cite{gatherallocal}.
We emphasize that Theorem \ref{thm:xi.bound} and Corollary \ref{cor:u.bound} do \emph{not} imply that $|\xi - \xib_m|$ and $|u - \ub_m|$ blow up as $m \to \infty$ for $\tau \geq 1$.  Indeed, in \cite{lorig-pagliarani-pascucci-2}, various numerical tests show that the expansions described in Proposition \ref{thm:xi} provide a high degree of accuracy for options with maturities of up to 10 years.
\end{remark}

%
%

\section{Examples}
\label{sec:examples}
In this section we implement our indifference pricing and implied volatility
approximations in two examples.  First, we consider an call options written on a traded asset.  Then we consider a European-style option on a non-traded asset.  In both examples, we take $t=0$ and fix the expansion point of the Taylor series to be the initial point of the diffusion $(\xb,\yb)=(X_0,Y_0)=(x,y)$.

\subsection{Heston: implied volatilities}
\label{sec:heston}
In our first example, we consider a stochastic volatility model, which, under the physical measure $\Pb$, is modeled by the following SDE:
\begin{align}
\begin{aligned}
\dd X_t
	&=	\Big( \lam(X_t,Y_t)\sqrt{Y_t} - \tfrac{1}{2} Y_t \Big)  \dd t + \sqrt{Y_t} \dd B_t^X , \\
\dd Y_t
	&=	\Big( \kappa (\theta - Y_t) + \rho \del \lam(X_t,Y_t) \sqrt{Y_t} \Big) \dd t 
				+ \del  \sqrt{Y_t} \( \rho  \dd B_t^X + \sqrt{1-\rho^2}  \dd B_t^Y \) .
\end{aligned}
&&\text{(under $\Pb$)} \label{eq:heston.P}
\end{align}
Comparing \eqref{eq:heston.P} with \eqref{eq:physical}, we identify
\begin{align}
\mu(x,y)
	&=	\lam(x,y) \sqrt{y}, &
\sig(y)
	&=	\sqrt{y}, \\
c(x,y)
	&=	\kappa (\theta - y) + \rho \del \lam(x,y) \sqrt{y} ,	&
\beta(y)
	&=	\del  \sqrt{y} .
\end{align}
Note that we have parametrized the drift function $\mu(x,y)$ via the volatility function 
$\sig(y)=\sqrt{y}$ 
and the Sharpe ratio $\lam(x,y)$, which we have left unspecified.  We have also included in $c(x,y)$ a term 
$\rho \del \lam(x,y) \sqrt{y} = \rho \beta(y) \lam(x,y)$. 
We do this primarily for reasons of computational convenience as will become more clear below.
\par
Under the minimal martingale measure $\Pbt$, defined via the Girsanov change of measure \eqref{eq:girsanov}, the dynamics of $(X,Y)$ are described by the following SDE:
\begin{align}
\begin{aligned}
\dd X_t
	&=	- \tfrac{1}{2} Y_t  \dd t + \sqrt{Y_t} \dd \Bt_t^X , \\
\dd Y_t
	&=	\kappa(\theta - Y_t) \dd t + \del  \sqrt{Y_t} \( \rho  \dd \Bt_t^X + \sqrt{1-\rho^2}  \dd \Bt_t^Y \)
\end{aligned}
&&\text{(under $\Pbt$)} \label{eq:heston.Q}
\end{align}
Note that, dynamics \eqref{eq:heston.Q} correspond to the model of \cite{heston1993}.  
Had we not included the term 
$\rho \del \lam(x,y) \sqrt{y}$ 
within the function $c(x,y)$ then the $\Pbt$ dynamics of $(X,Y)$ would not have corresponded to Heston.
\subsubsection*{Implied volatility}
The functions that play a key role in the implied volatility expansion (Proposition \ref{thm:imp.vol}) are
\begin{align}
a &:= \tfrac{1}{2} \sig^2
	=	 \tfrac{1}{2} y, &
b	&:=	\tfrac{1}{2} \beta^2
	=	\tfrac{1}{2} \del^2 y, \\
f	&:=	c - \rho \beta \lam 
	=	\kappa(\theta - y) , &
g &:= \rho \sig \beta
	=	\rho \del y .
\end{align}
Due to the manner in which we specified the $\Pb$-dynamics of $(X,Y)$, the Sharpe ratio $\lam$ does not appear in $f$.
Thus, the effect of the Sharpe ratio $\lam$ on the second order approximation of implied volatility $\IIb_2$ is felt only through the term $\II_2^\Ind$ \eqref{eq:I2.indiff}.
\par
We see from equation \eqref{eq:I2.indiff} that the first term in $\II_2^\Ind$ is linear in $(\tfrac{1}{2}\lam^2)_{0,1}$.
Moreover, since 
$(\tfrac{1}{2}\sig^2)_{0,1}=\d_y (\tfrac{1}{2}\sig^2(y))=\tfrac{1}{2}>0$, 
the first term in $\II_2^\Ind$ has the sign of $-(\tfrac{1}{2}\lam^2)_{0,1}$.  We observe also, that the first term in $\II_2^\Ind$ is independent of $k$.  Thus, increasing $(\tfrac{1}{2}\lam^2)_{0,1}$ shifts the approximation $\IIb_2$ of the implied volatility surface downward (for both the buyer and the seller and for all strikes and maturities).  Decreasing $(\tfrac{1}{2}\lam^2)_{0,1}$ raises the $\IIb_2$ for all strikes and maturities.
\par
On the left-hand side of Figure \ref{fig:heston}, we plot the approximate implied volatilities $\bar{\II}_2$ of both the buyer and the seller.  
For comparison, we also plot the exact and second order approximation of implied volatility corresponding to the Heston price, which is defined as
\begin{align}
q(0,x,y)
	&:=	\Ebt_{0,x,y} (\ee^{X_T} - \ee^k)^+ ,
\end{align}
where the dynamics of $(X,Y)$ under $\Pbt$ are given by \eqref{eq:heston.Q}.
Exact implied volatilities are obtained by computing the exact prices in the Heston model and then inverting the Black-Scholes formula numerically.
Second order approximate implied volatilities for the Heston model are obtained by removing the nonlinear term from the second order approximation of the implied volatility corresponding to the indifference price: $\IIb_2 - \II_2^\Ind$.
\par
On the right-hand side of Figure \ref{fig:heston}, we plot $|\II_2^\Ind|$.  We assume in both plots that $(\tfrac{1}{2}\lam^2)_{0,1}=0$ since, as previously mentioned, this terms simply shifts the buyer and seller implied volatility curves vertically.  Under this assumption, $|\II_2^\Ind|$ is an approximation for one half the bid-ask spread.  We observe that the maximum of $|\II_2^\Ind|$ occurs at $k>x$ and increases with increasing maturity $T$.

\subsection{Reciprocal Heston model: indifference prices}
\label{sec:vol}
We consider now a second Stochastic volatility model $(X,Y)$, which, under the physical measure $\Pb$, is modeled by the following SDE:
\begin{align}
\begin{aligned}
\dd X_t
	&=	\( \mu - \tfrac{1}{2} Y_t \) \dd t + \sqrt{Y_t} \dd B_t^X , \\
\dd Y_t
	&=	\( a Y_t + \frac{2 (b^2 - a \kappa)}{\mu^2(1-\rho)^2}Y_t^2 \) \dd t 
			- \( \frac{2}{1-\rho^2} \)^{1/2} \frac{b}{\mu} Y_t^{3/2}\( \rho  \dd B_t^X + \sqrt{1-\rho^2}  \dd B_t^Y \) .
\end{aligned} 
&&\text{(under $\Pb$)}
\label{eq:XY}
\end{align}
The model is referred to as \emph{reciprocal Heston} since $Y$ is the reciprocal of a CIR process
\begin{align}
Y_t
	&=	\frac{\mu^2(1-\rho^2)}{2 R_t} , &
\dd R_t
	&=	a ( \kappa - R_t ) \dd t + b \sqrt{R_t} \( \rho \dd B_t^X + \sqrt{1-\rho^2}  \dd B_t^Y \) .
\end{align}
Here, $(a,b,\kappa)$ must satisfy the usual Feller condition: $2 a \kappa \geq b^2$.
Comparing \eqref{eq:XY} with \eqref{eq:SV}, we identify
\begin{align}
\mu(y)
	&=	\mu , &
\sig(y)
	&=	\sqrt{y}, \\
c(y)
	&=	a y + \frac{2 (b^2 - a \kappa)}{\mu^2(1-\rho)^2}y^2 , &
\beta(y)
	&=	- \( \frac{2}{1-\rho^2} \)^{1/2} \frac{b}{\mu} y^{3/2} .
\end{align}
\subsubsection*{Indifference prices}
We will compute indifference prices for a European-style claim whose payoff $\varphi(Y_T)$ depends only on the terminal value of the stochastic variance process $Y$.  Although derivatives on the terminal value of variance do not actively trade, this is a useful model in which to test the accuracy of our pricing approximation, since exact indifference prices have been computed in \cite{grasselli-hurd-2007}.  The $m$th-order approximate indifference price $\ub_m$ can be computed using Proposition \ref{thm:xi} and equations \eqref{eq:psi0.again}, \eqref{eq:psik.again} and \eqref{eq:ubar.m}.
\par
As pointed out by \cite{grasselli-hurd-2007}, unbounded payoffs results in an expected utility of negative infinity.  As such, we will focus on bounded payoffs.  In particular, we consider a payoff $\varphi$ that is the difference of call option payoffs:
\begin{align}
\varphi(y)
	&=	(y - k_1)^+ - (y - k_2)^+ , &
k_1
	&<	k_2 , \label{eq:payoff}
\end{align}
which are bounded by $(k_2 - k_1)$.  In Figures \ref{fig:VolCall} and \ref{fig:VolCall2} we plot the exact and approximate buyer's and seller's indifference prices, for series of strikes $k_1$ with $k_2$ fixed.  We also plot the zeroth-, first-, and second-order indifference price approximations.  The plots clearly show that the second order approximation of the indifference price $\ub_2$, nearly coincides with the exact indifference price $u$.

%
%

\section{Conclusion}
In this paper, under a general class of LSV dynamics, we derive an explicit approximation for the indifference price of European-style asset, whose payoff may depend on either a traded or non-traded asset.   For call options on a traded asset, we translate the price approximation into an explicit approximation of implied volatility.  For options on a non-traded asset, we derive rigorous error bounds for the indifference price approximation.  
Lastly, we implement our indifference price and implied volatility approximations in two examples.

\subsection*{Acknowledgments}
The author would like to thank Jean-Pierre Fouque, Tim Leung, Stefano Pagliarani, Andrea Pascucci, Ronnie Sircar and Stephan Sturm for a number of fruitful discussions.
Additionally, the author would like to extend his gratitude to two anonymous reviewers and one anonymous associate editor, whose comments helped improve both the quality and readability of this manuscript.

%
%


\appendix
\section{Proof of Proposition \ref{thm:phi}}
\label{sec:phi.proof}
In this appendix we compute $\eta_0$, $\eta_1$ and $\eta_2$.  The function $\eta_0$ satisfies ODE \eqref{eq:psi0.pde.2}.  The explicit solution, which can be easily checked, is
\begin{align}
\eta_0(t)
	&=	- (T-t) (\tfrac{1}{2}\lam^2)_0 .
\end{align}
Next, we compute $\eta_1$.  Omitting arguments $(x,y)$ for clarity, we have
\begin{align}
\eta_1(t)
	&=	\int_t^T \dd t_1  \Pc_0(t,t_1) H_1(t_1) &
	&		\text{(by \eqref{eq:psin.pde} and \eqref{eq:duhamel})} \\
	&=	\int_t^T \dd t_1  \Pc_0(t,t_1) \( \Act_1 \eta_0(t_1) - (\tfrac{1}{2}\lam^2)_1 \) 1 &
	&		\text{(by \eqref{eq:F1})} \label{eq:eta1.ref} \\
	&=	- \int_t^T \dd t_1  (\tfrac{1}{2}\lam^2)_1(\Xc(t,t_1),\Yc(t,t_1)) 1 . &
	&		\text{(by \eqref{eq:Pc.poly})} , \label{eq:phi1} 
\end{align}
which is the expression given in \eqref{eq:eta1.theorem}.
Note, in the second-to-last equality we have used $\Act_1 \eta_0 = 0$ since $\eta_0$ is independent of $(x,y)$. 
\par
Finally, we compute $\eta_2$.  Once again, omitting the arguments $(x,y)$ for clarity, we have
\begin{align}
&\eta_2(t)
	 =	\int_t^T \dd t_1  \Pc_0(t,t_1) H_2(t_1) &
	&		\text{(by \eqref{eq:psin.pde} and \eqref{eq:duhamel})} \\
	&=	\int_t^T \dd t_1  \Pc_0(t,t_1) 
			\Big( \Act_2 \eta_0(t_1)  + \Act_1 \eta_1(t_1) \Big) \\ &\quad 
			+ \int_t^T \dd t_1  \Pc_0(t,t_1) 
			\Big(	- (\tfrac{1}{2}\lam^2)_2 + (1-\rho^2) ( \tfrac{1}{2} \beta^2)_0 (\d_y \eta_1(t_1))^2 \Big) &
	&		\text{(by \eqref{eq:F2})} \\
	&=	- \int_t^T \dd t_1 \int_{t_1}^T \dd t_2  \Gc_1(t,t_1) \Pc_0(t,t_1) \Pc_0(t_1,t_2) (\tfrac{1}{2}\lam^2)_1 &
	&		\text{(by \eqref{eq:eta1.ref} and \eqref{eq:PA=GP})} \\ &\quad
			- \int_t^T \dd t_1  (\tfrac{1}{2}\lam^2)_2 (\Xc(t,t_1),\Yc(t,t_1)) 1
			+ (1-\rho^2) ( \tfrac{1}{2} \beta^2)_0 \int_t^T \dd t_1  (\tfrac{1}{2}\lam^2)_{0,1}^2 (T-t_1)^2 &
	&		\text{(by \eqref{eq:Pc.poly} and \eqref{eq:eta1})} \\ 
	&=	- \int_t^T \dd t_1 \int_{t_1}^T \dd t_2  \Gc_1(t,t_1) \Pc_0(t,t_2) (\tfrac{1}{2}\lam^2)_1 \\ &\quad
			- \int_t^T \dd t_1  (\tfrac{1}{2}\lam^2)_2 (\Xc(t,t_1),\Yc(t,t_1))  1
			+ \frac{1}{3}(T-t)^3 (1-\rho^2) ( \tfrac{1}{2} \beta^2)_0 (\tfrac{1}{2}\lam^2)_{0,1}^2 &
	&		\text{(semigroup property)} \\
	&=	- \int_t^T \dd t_1  \Gc_1(t,t_1) \int_{t_1}^T \dd t_2 (\tfrac{1}{2}\lam^2)_1(\Xc(t,t_2),\Yc(t,t_2)) 1 \\ &\quad
			- \int_t^T \dd t_1  (\tfrac{1}{2}\lam^2)_2 (\Xc(t,t_1),\Yc(t,t_1)) 1
			+ \frac{1}{3}(T-t)^3 (1-\rho^2) ( \tfrac{1}{2} \beta^2)_0 (\tfrac{1}{2}\lam^2)_{0,1}^2 , &
	&		\text{(by \eqref{eq:Pc.poly})} 
\end{align}
which is the expression given in \eqref{eq:eta2.theorem}.  This proves Proposition \ref{thm:phi}.


\section{Proof of Proposition \ref{thm:u}}
\label{sec:u.proof}
In this appendix we derive explicit expressions for $u_0$, $u_1$ and $u_2$.
As always, throughout this appendix we will suppress $(x,y)$-dependence, except where it is needed for clarity.
We begin with $u_0$, the unique classical solution of \eqref{eq:eta0.pde.2}.  Using \eqref{eq:duhamel}, we can immediately write 
\begin{align}
u_0(t)
	&=	\Pc_0(t,T) \varphi . \label{eq:eta0}
\end{align}
In particular, for call payoffs $\varphi(x) = (\ee^x - \ee^k)^+$ and put payoffs $\varphi(x) = (\ee^k - \ee^x)^+$, expression \eqref{eq:eta0} becomes $u_0(t) = u^\BS(t,\cdot;\sig_0)$, where $u^\BS$ is given in \eqref{eq:uBS}.  This is precisely the expression given in \eqref{eq:u0.theorem} for $u_0$.
\par
Next, we compute the function $u_1$.  We have
\begin{align}
u_1(t)
	&=	\int_t^T \dd t_1  \Pc_0(t,t_1) U_1(t_1) &
	&		\text{(by \eqref{eq:etan.pde} and \eqref{eq:duhamel})} \\
	&=	\int_t^T \dd t_1  \Pc_0(t,t_1) \Act_1 u_0(t_1) &
	&		\text{(by \eqref{eq:G1})}  \label{eq:u1.intermediate} \\
	&=	\int_t^T \dd t_1  \Gc_1(t,t_1) \Pc_0(t,t_1) \Pc_0(t_1,T) \varphi &
	&		\text{(by \eqref{eq:PA=GP} and \eqref{eq:eta0})}	\\
	&=	\Big( \int_t^T \dd t_1  \Gc_1(t,t_1) \Big) u_0(t) &
	&		\text{(by \eqref{eq:semigroup}) and \eqref{eq:eta0}} . \label{eq:u1}
\end{align}
which is the expression given for $u_1$ in \eqref{eq:u1.theorem}.
\par
Finally, for $u_2$, from \eqref{eq:etan.pde}, \eqref{eq:G2} and \eqref{eq:duhamel} we have
\begin{align}
u_2(t)
	&=	\int_t^T \dd t_1  \Pc_0(t,t_1) U_2(t_1) 
	=	u_{2,1}(t) + (1- \rho^2) (\tfrac{1}{2} \beta^2)_0  \Big(2  u_{2,2}(t) -  \gam \nu   u_{2,3}(t) \Big), \label{eq:u2=u21u22u23}
\end{align}
where we have defined
\begin{align}
u_{2,1}(t)
	&=	\int_t^T \dd t_1  \Pc_0(t,t_1) \Big( \Act_2 u_0(t_1) + \Act_1 u_1(t_1) \Big), \label{eq:u21} \\
u_{2,2}(t)
	&=	\int_t^T \dd t_1  \Pc_0(t,t_1) 
			(\d_y u_1(t_1) )(\d_y \eta_1(t_1)) , \label{eq:u22} \\
u_{2,3}(t)
	&= \int_t^T \dd t_1  \Pc_0(t,t_1) 
			\(\d_y u_1(t_1) \)^2 . \label{eq:u23}
\end{align}
We will analyze $u_{2,1}$, $u_{2,2}$ and $u_{2,3}$ individually starting with $u_{2,1}$. We have
\begin{align}
u_{2,1}(t)
	&=	\int_t^T \dd t_1  \Pc_0(t,t_1) \Act_2 u_0(t_1) \\ &\quad 
			+ \int_t^T \dd t_1 \int_{t_1}^T \dd t_2  \Pc_0(t,t_1) \Act_1 \Pc_0(t_1,t_2) \Act_1 u_0(t_2) &
	&		\text{(by \eqref{eq:u1.intermediate} and \eqref{eq:u21})} \\ 
	&=	\int_t^T \dd t_1  \Gc_2(t,t_1) \Pc_0(t,t_1) \Pc_0(t_1,T) \varphi \\ &\quad 
			+ \int_t^T \dd t_1 \int_{t_1}^T \dd t_2  \Gc_1(t,t_1) \Gc_1(t,t_2) \Pc_0(t,t_2) \Pc_0(t_2,T) \varphi &
	&		\text{(by \eqref{eq:PA=GP} and \eqref{eq:eta0})} \\
	&=	\Big( \int_t^T \dd t_1  \Gc_2(t,t_1) +
			\int_t^T \dd t_1 \int_{t_1}^T \dd t_2  \Gc_1(t,t_1) \Gc_1(t,t_2) \Big) u_0(t) &
	&		\text{(by \eqref{eq:eta0})} \label{eq:u21.explicit}
\end{align}
Next, we analyze $u_{2,2}$.  
From \eqref{eq:eta1}, a direct computation yields 
\begin{align}
\d_y \eta_1(t)
	&=	- (\tfrac{1}{2}\lam^2)_{0,1} (T-t) . \label{eq:dy.phi1}
\end{align}
Since the function $\d_y \eta_1$ in independent of $(x,y)$, we have 
\begin{align}
\Pc_0(t,t_1) (\d_y u_1(t_1) )(\d_y \eta_1(t_1)) = (\d_y \eta_1(t_1)) \Pc_0(t,t_1) (\d_y u_1(t_1) ) , \label{eq:Pphi=phiP}
\end{align}
and thus, we focus on computing $\Pc_0(t,t_1) (\d_y u_1(t_1) )$.  To this end, we observe that the semigroup operator commutes with constant coefficient differential operators
\begin{align}
\d_y^n \d_x^m \Pc_0(t,t_1) 
	&= \Pc_0(t,t_1) \d_y^n \d_x^m . \label{eq:dP=Pd} 
\end{align}
This is easily proved using integration by parts and symmetry properties of $\Gam_0$.  Using this commutation properly, we compute
\begin{align}
\Pc_0(t,t_1) \(\d_y u_1(t_1)\)
	&=	\d_y \Pc_0(t,t_1) u_1(t_1) & 
	&		\text{(by \eqref{eq:dP=Pd})} \\
	&=	\d_y \Pc_0(t,t_1) \int_{t_1}^T \dd t_2  \Pc_0(t_1,t_2) \Act_1 u_0(t_2) & 
	&		\text{(by \eqref{eq:u1.intermediate})} \\
	&=	\d_y \int_{t_1}^T \dd t_2   \Pc_0(t,t_2) \Act_1 u_0(t_2) &
	&		\text{(by \eqref{eq:semigroup})} \\
	&=	\d_y \int_{t_1}^T \dd t_2  \Gc_1(t,t_2) \Pc_0(t,t_2) \Pc(t_2,T) \varphi &
	&		\text{(by \eqref{eq:PA=GP} and \eqref{eq:eta0})} \\
	&=	\d_y \int_{t_1}^T \dd t_2   \Gc_1(t,t_2) u_0(t) . &
	&		\text{(by \eqref{eq:semigroup} and \eqref{eq:eta0})} \label{eq:P0dyu1}
\end{align}
Thus, we have 
\begin{align}
u_{2,2}(t)
	&=	\int_t^T \dd t_1  (\d_y \eta_1(t_1)) \Pc_0(t,t_1) (\d_y u_1(t_1) ) & 
	&		\text{(by \eqref{eq:u23} and \eqref{eq:Pphi=phiP})} \\
	&=	\int_t^T \dd t_1  (\d_y \eta_1(t_1)) \d_y \int_{t_1}^T \dd t_2   \Gc_1(t,t_2) u_0(t) &
	&		\text{(by \eqref{eq:P0dyu1})} \\
	&=	\( \int_t^T \dd t_1 \int_{t_1}^T \dd t_2  (\d_y \eta_1(t_1)) \cdot \d_y \Gc_1(t,t_2) \) u_0(t) . \label{eq:u22.explicit}
\end{align}
where the function $(\d_y \eta_1(t_1))$ is given by \eqref{eq:dy.phi1}.
\par
Lastly, we compute $u_{2,3}$.  First, using \eqref{eq:dy.G1} and \eqref{eq:u1}, a direct computation shows that 
\begin{align}
\d_y u_1(t_1)
	&=	(T-t_1) (\tfrac{1}{2}\sig^2)_{0,1} (\d_x^2 - \d_x ) u_0(t_1) .
\end{align}
Therefore, we have
\begin{align}
\Pc_0(t,t_1)  \( \d_y u_1(t_1) \)^2
	&=	(T-t_1)^2 (\tfrac{1}{2}\sig^2)_{0,1}^2 \Pc_0(t,t_1) \( (\d_x^2 - \d_x ) u_0(t_1) \)^2 . \label{eq:mid-step}
\end{align}
Next, using $u_0 = u^\BS(t,\cdot;\sig_0)$ and the explicit expression for the Black-Scholes price \eqref{eq:uBS}, we obtain
\begin{align}
\( (\d_x^2 - \d_x ) u_0(t_1) \)^2
	&=	\frac{1}{\sig_0^2 (T-t_1)}\ee^{2x} \phi^2( d_+(t_1,x;\sig_0) ) , &
\phi
	&=	\Phi' ,
\end{align}
where we have introduced $\phi$, the density of a standard normal random variable.
Hence, \eqref{eq:mid-step} becomes
\begin{align}
\Pc_0(t,t_1)  \( \d_y u_1(t_1) \)^2
	&=	(T-t_1) \frac{(\tfrac{1}{2}\sig^2)_{0,1}^2 }{\sig_0^2} \Pc_0(t,t_1) \ee^{2x} \phi^2(d_+(t_1,x;\sig_0) ) . \label{eq:P0.dy.u1.2}
\end{align}
Now, we compute
\begin{align}
\Pc_0(t,t_1) \ee^{2x} \phi^2(d_+(t_1,x;\sig_0) )
	&=	\int_\Rb \dd x_1 \frac{\exp\(2 x_1 - d_+^2(t_1,x_1;\sig_0) \)}{2\pi\sqrt{2 \pi \sig_0^2 (t_1-t)}}
			\exp\( \frac{-(x_1-x+\tfrac{1}{2}\sig_0^2(t_1-t))^2}{2 \sig_0^2 (t_1-t)}\) \\
	&=	\frac{1}{2\pi \sqrt{2 \pi \sig_0^2 (t_1-t)}} \int_\Rb \dd x_1  \exp \( -a x_1^2 + b x_1 + c \) ,
\end{align}
where $a$, $b$ and $c$ are given by
\begin{align}
a
	&=	\frac{1}{\sig_0^2} \( \frac{1}{(T-t_1)} + \frac{1}{2(t_1-t)} \), \\
b
	&=	\frac{1}{2}+\frac{1}{\sig_0^2}\(\frac{2 k}{T-t_1}+\frac{2 x}{2 (t_1-t)} \) , \\
c
	&=  -\frac{1}{\sig_0^2}\(\frac{\(k-\frac{1}{2} \sigma ^2 (T-t_1)\)^2}{T-t_1}+\frac{\(x-\frac{1}{2} \sig_0^2 (t_1-t)\)^2}{2 (t_1-t)}\) .
\end{align}
Using
\begin{align}
\int_\Rb \dd x \exp \( -a x^2 + b x + c \)
	&=	\frac{\sqrt{\pi }}{\sqrt{a}} \exp\( \frac{b^2}{4 a}+c \) , & 
a 
	&>	0 , \label{eq:abc}
\end{align}
we obtain
\begin{align}
\Pc_0(t,t_1) \ee^{2x} \phi^2(d_+(t_1,x;\sig_0) )
	&=	\frac{1}{2\pi}\sqrt{\frac{T-t_1}{T-t+t_1-t}}
			\exp \( 2 k-\frac{\left((k-x)+\frac{1}{2} \sig_0^2 (T-t)\right)^2}{\sig_0^2 (T-t+t_1-t)} \) \label{eq:P0.phi.2}
\end{align}
Finally, from \eqref{eq:u22} we have
\begin{align}
u_{2,3}(t)
	&=	\int_t^T \dd t_1  \Pc_0(t,t_1) 
			\(\d_y u_1(t_1) \)^2 \\
	&=	\frac{(\tfrac{1}{2}\sig^2)_{0,1}^2 }{\sig_0^2} \int_t^T \dd t_1  (T-t_1) \Pc_0(t,t_1) \ee^{2x} \phi^2(d_+(t_1,x;\sig_0) ) &
	&		\text{(by \eqref{eq:P0.dy.u1.2})} \\
	&=	\frac{(\tfrac{1}{2}\sig^2)_{0,1}^2 }{2\pi\sig_0^2} \int_t^T \dd t_1  \frac{(T-t_1)^{3/2}}{\sqrt{T-t+t_1-t}}
			\exp \( 2 k-\frac{\((k-x)+\frac{1}{2} \sig_0^2 (T-t)\)^2}{\sig_0^2 (T-t+t_1-t)} \) . &
	&		\text{(by \eqref{eq:P0.phi.2})} \label{eq:u23.explicit}
\end{align}
Combining expressions \eqref{eq:u2=u21u22u23} with \eqref{eq:u21.explicit}, \eqref{eq:u22.explicit} and \eqref{eq:u23.explicit} yields the expression given for $u_2$ in \eqref{eq:u2.theorem}.  This concludes the proof of Theorem \ref{thm:u}.


\section{Proof of Proposition \ref{thm:imp.vol}}
\label{sec:imp.vol.proof}
In this appendix, we establish the formulas provided in Proposition \ref{thm:imp.vol}.  We remind the reader that $(t,x,y,k,T)$ are fixed throughout and we write these arguments only when needed for clarity.
We start by observing from \eqref{eq:u1.theorem} and \eqref{eq:u2.0} that $u_1$ and $u_2^0$ can be written as
\begin{align}
u_1
	&=	\Big( \int_t^T \dd t_1  \Gct_1(t,t_1) \Big) (\d_x^2 - \d_x )u^\BS(\sig_0) , \label{eq:u1.again} \\
u_2^0
	&=	\Big( \int_t^T \dd t_1  \Gct_2(t,t_1) + \int_t^T \dd t_1 \int_{t_1}^T \dd t_2  \Gc_1(t,t_1)\Gct_1(t,t_2) \Big) (\d_x^2 - \d_x )
			u^\BS(\sig_0) , \label{eq:u2.again}
\end{align}
where the operator $\Gct_i(t,t_1)$ is given by
\begin{align}
\Gct_i(t,t_1)
	&=	(\tfrac{1}{2}\sig^2)_i(\Xc(t,t_1),\Yc(t,t_1)) .
\end{align}
Thus, it is clear that $u_1$ and $u_2^0$ are finite sums of the form
\begin{align}
u_1
	&=	\sum_n C_n \d_x^n (\d_x^2 - \d_x) u^\BS(\sig_0) , &
u_2^0
	&=	\sum_n C_n^0 \d_x^n (\d_x^2 - \d_x) u^\BS(\sig_0) , \label{eq:sums}
\end{align}
where the coefficients $C_n$ and $C_n^0$ depend on $(t,x,y)$ and can be computed from \eqref{eq:u1.again} and \eqref{eq:u2.again}, respectively.  Because of the number of terms involved, the coefficients are best computed using a computer algebra program such as Mathematica.  Next, using \eqref{eq:uBS}, a direct computation shows
\begin{align}
(\d_x^2 - \d_x ) u^\BS(\sig_0)
    &=  \frac{1}{\sig_0\sqrt{2\pi\tau}}\ee^{-z^2+k} , &
z
    &:= \frac{x-k-\frac{1}{2}\sig_0^2 \tau}{\sig_0\sqrt{2\tau}} , &
\tau
    &:= T-t . \label{eq:result}
\end{align}
Hence
\begin{align}
\frac{\d_x^n (\d_x^2 - \d_x ) u^{\BS}(\sig_0)}{(\d_x^2 - \d_x ) u^\BS(\sig_0)}
    &=  \ee^{z^2}\d_x^n \ee^{-z^2}
    =       \(\frac{1}{\sig_0\sqrt{2\tau}}\)^n\ee^{z^2}\d_z^n \ee^{-z^2}
    =       \(\frac{-1}{\sig_0\sqrt{2\tau}}\)^n h_n(z) , \label{eq:fraction}
\end{align}
where $h_n(z) := (-1)^n \ee^{z^2} \d_z^n \ee^{-z^2}$ is the $n$-th Hermite polynomial.  Combining \eqref{eq:sums} and \eqref{eq:fraction} with the classical Vega-Gamma relation 
$\d_\sig u^\BS(\sig_0) =	\sig_0 \tau (\d_x^2 - \d_x) u^\BS(\sig_0)$,
we see that
\begin{align}
\frac{u_1}{\d_\sig u^\BS(\sig_0)}
	&=	\frac{1}{\sig_0 \tau } \sum_n C_n \(\frac{-1}{\sig_0\sqrt{2\tau}}\)^n h_n(z) , &
\frac{u_2^0}{\d_\sig u^\BS(\sig_0)}
	&=	\frac{1}{\sig_0 \tau } \sum_n C_n^0 \(\frac{-1}{\sig_0\sqrt{2\tau}}\)^n h_n(z) . \label{eq:sums.2}
\end{align}
We also have from \eqref{eq:uBS}
\begin{align}
\frac{\d_\sig^2 u^\BS(\sig_0)}{\d_\sig^2 u^\BS(\sig_0)}
	&=	\frac{(k-x)^2}{\sig_0^3 \tau}-\frac{\sig_0 \tau }{4} . \label{eq:d2sig}
\end{align}
Combining \eqref{eq:I2} with \eqref{eq:sums.2} and \eqref{eq:d2sig} one obtains the expressions given in Proposition \ref{thm:imp.vol} for $\II_{1,0}$, $\II_{0,1}$, $\II_{2,0}$, $\II_{1,1}$ and $\II_{0,2}$.  What remains is to compute $u_2^\Ind / \d_\sig u^\BS(\sig_0)$.  First, we have 
\begin{align}
&\frac{1}{\d_\sig u^\BS(\sig_0)}\Big( \int_t^T \dd t_1 \int_{t_1}^T \dd t_2  (\d_y \eta_1(t_1)) \cdot \d_y \Gc_1(t,t_2) \Big) u^\BS(\sig_0) \\
	&=	\frac{- 1}{\tau \sig_0 (\d_x^2 - \d_x)u^\BS(\sig_0) }\Big( \int_t^T \dd t_1 \int_{t_1}^T \dd t_2  (\tfrac{1}{2}\lam^2)_{0,1} (T-t_1) ( \tfrac{1}{2} \sig^2)_{0,1} \Big) (\d_x^2 - \d_x )u^\BS(\sig_0) \\
	&=	\frac{- (\tfrac{1}{2}\lam^2)_{0,1} ( \tfrac{1}{2} \sig^2)_{0,1} \tau^2 }{3 \sig_0} .
\end{align}
Next, an explicit computation gives 
\begin{align}
\frac{ \ee^{2k} }{\d_\sig u^\BS(\sig_0)}
	&=	\frac{\ee^k\sqrt{2\pi}}{\sqrt{\tau}} \exp \( \frac{1}{2 \sig_0^2 \tau } \( (k-x) + \frac{1}{2}\sig_0^2 \tau  \)^2 \) .
\end{align}
Combining the above results with expression \eqref{eq:u2.indiff.2} for $u_2^\Ind$, we obtain expression \eqref{eq:I2.indiff}.

%
%

\begin{small}
\bibliographystyle{chicago}
\bibliography{Bibtex-Master-3.03}
\end{small}

%
%

\clearpage

\begin{SCfigure}
\centering
\caption{
We plot $\eps$ on $\chi^\eps(x) := \chi(\xb + \eps( x - \xb))$ as a function of $x$ with $\eps = 0$ (dotted line), $\eps = 1/4$ (dot-dashed line), $\eps = 1/2$ (dashed line) and $\eps = 1$ (solid line).
In this figure we take $\chi(x) = \arctan x + \pi/2$ and fix $\xb = 0$.
Note that $\chi^\eps(x)|_{\eps=0} = \chi(\xb)$ is a constant function and $\chi^\eps(x)|_{\eps=1} = \chi(x)$.
The smaller the value of $\eps$ the less the function $\chi^\eps(x)$ varies with $x$.
}
\includegraphics[width=0.475\textwidth]{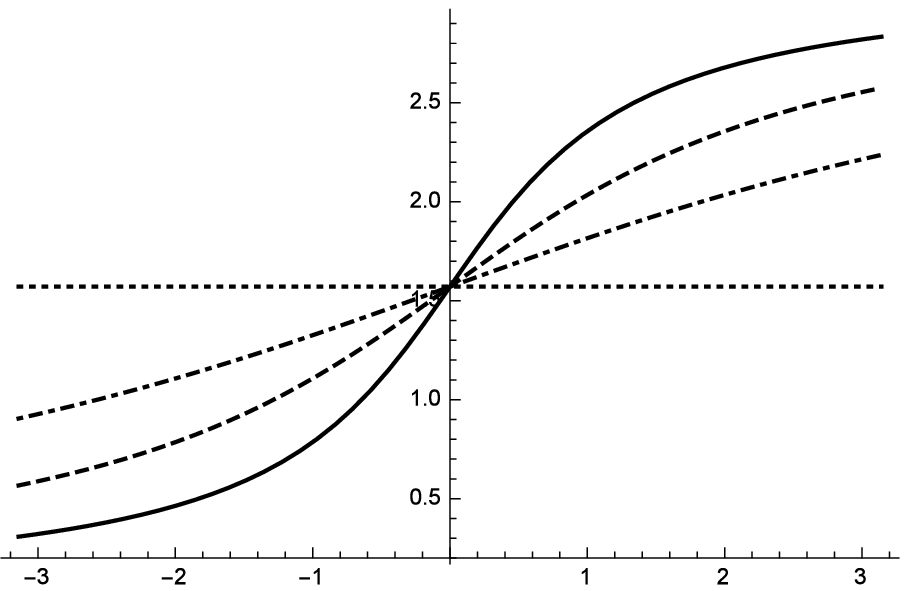}
\label{fig:eps}
\end{SCfigure}

\begin{figure}
\centering
\begin{tabular}{cc}
\includegraphics[width=0.475\textwidth]{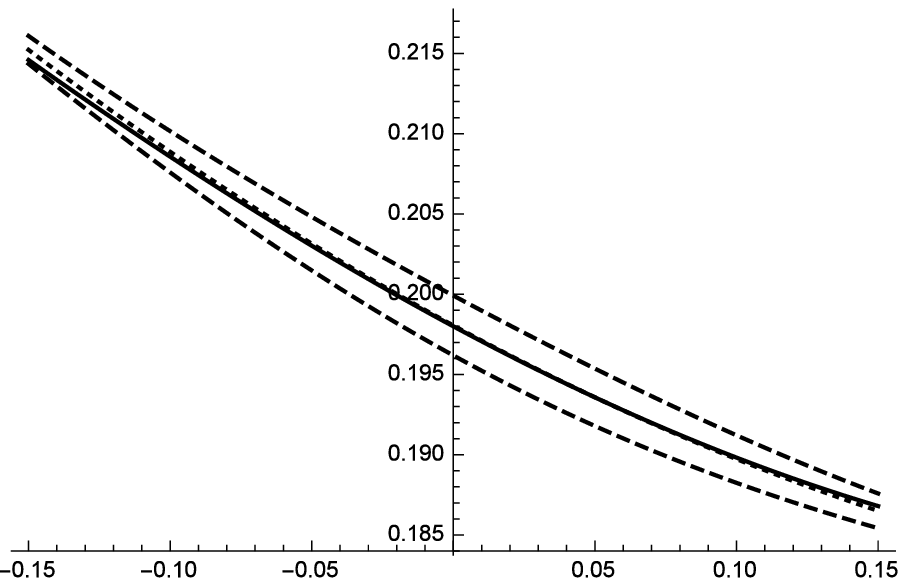}&
\includegraphics[width=0.475\textwidth]{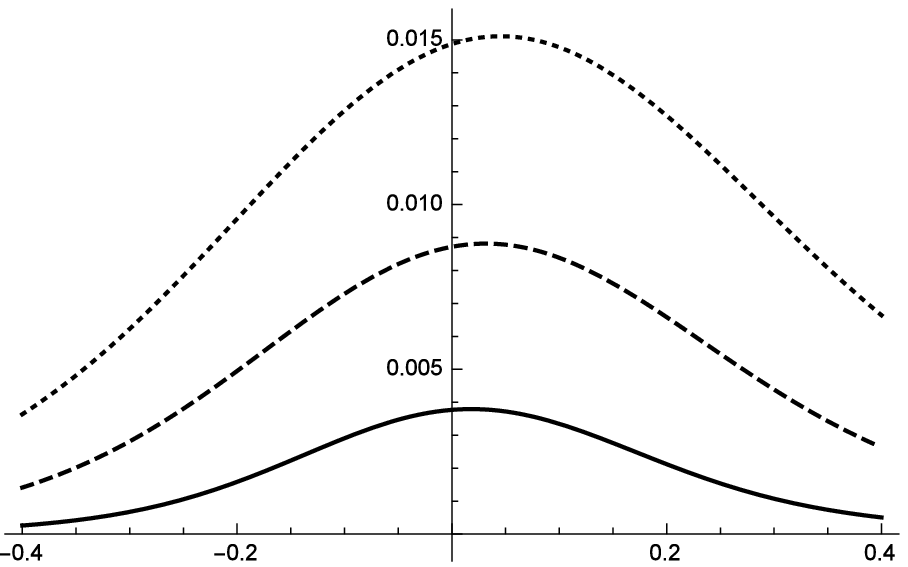}
\end{tabular}
\caption{
\emph{Left}:
Approximate implied volatilities $\bar{\II}_2$ generated by the buyer's indifference price (bottom dashed line) and seller's indifference price (top dashed line) are plotted as a function of $\log$ moneyness $L=k-x$ for the model considered in Section \ref{sec:heston}.  
For comparison, we also plot the exact implied volatility corresponding to the Heston model (solid line) and our second order approximation of this quantity (dotted line), which is given by $\IIb_2 - \II_2^\Ind$.  The maturity if fixed at $T=0.25$ years.
\emph{Right}: 
The absolute value $|\II_2^\Ind|$, an approximation of half the bid-ask-spread, is plotted as a function of $\log$ strike $k$ for three different maturities $T=\{0.3, 0.7, 1.0\}$, corresponding to the solid, dashed and dotted lines.
The following parameters are used in both plots:
$\del=0.2$, $\theta =0.04$, $\kappa =1.15$, $\rho =-0.4$, $y = \log \theta$, $x=0$,
$(\tfrac{1}{2}\lam^2)_{0,1}=0$, and
$\gam \nu = \pm 25$.
}
\label{fig:heston}
\end{figure}


\clearpage

\begin{figure}
\centering
\begin{tabular}{cc}
\includegraphics[width=0.475\textwidth]{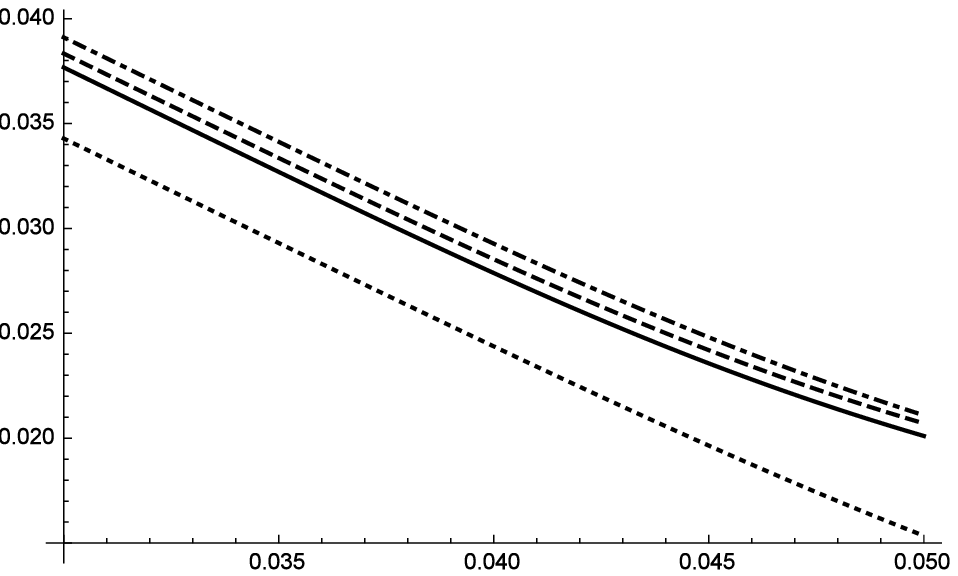}&
\includegraphics[width=0.475\textwidth]{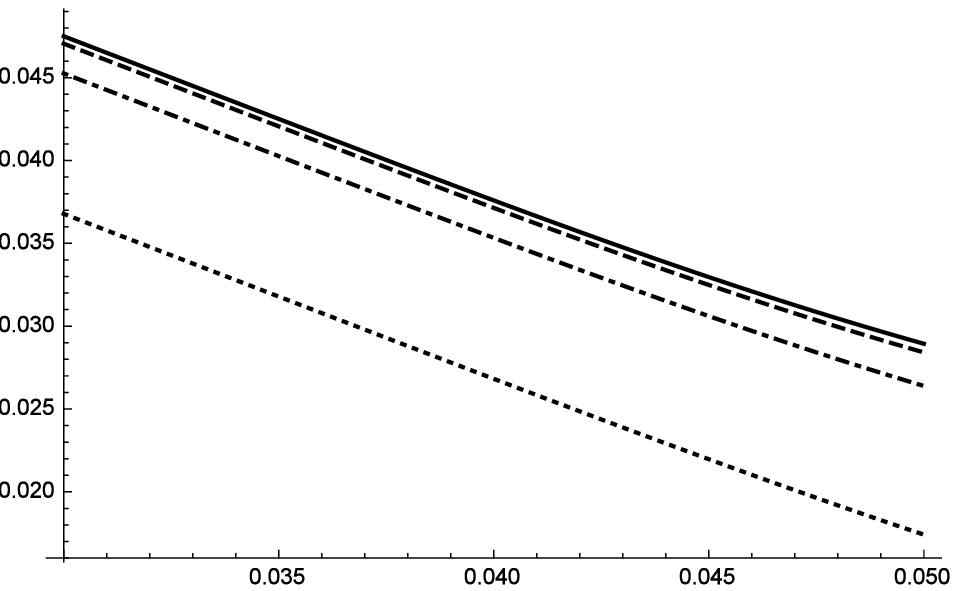}
\end{tabular}
\caption{
For the model described in Section \ref{sec:vol}, we consider an option on the terminal value of variance $Y_T$ with a payoff function $\varphi$ given by \eqref{eq:payoff}.  
\emph{Left}: 
We plot the buyer's indifference price $(\gam \nu = 40)$ as a function of $k_1$ with $k_2$ fixed.
\emph{Right}:
We plot the seller's indifference price $(\gam \nu = -25)$ as a function of $k_1$ with $k_2$ fixed.
In both plots, the dotted line corresponds to the zeroth order approximation  $\ub_0$,
the dash-dotted line corresponds to the first order approximation  $\ub_1$,
the dashed line corresponds to the second order approximation $\ub_2$,
and the solid line corresponds to the exact indifference price $u$.
The following parameters are used in both plots:
$a=5$, $b=0.04$, $\kappa=0.001$, $\mu=0.02$, $\rho=0.2$, $T=0.15$, $y=0.04$, $k_2=2.0$.
}
\label{fig:VolCall}
\end{figure}

\begin{SCfigure}
\centering
\caption{
For the model described in Section \ref{sec:vol}, we consider an option on the terminal value of variance $Y_T$ with a payoff function $\varphi$ given by \eqref{eq:payoff}.  
We plot the buyer's indifference price (bottom) and the seller's indifference price (top) as a function of $k_1$ with $k_2$ fixed.  The dashed lines corresponds to the second order approximation $\ub_2$,
and the solid lines corresponds to the exact indifference price $u$.
The following parameters are used throughout:
$a=5$, $b=0.04$, $\kappa=0.001$, $\mu=0.02$, $\rho=0.2$, $T=0.15$, $y=0.04$, $k_2=2.0$, $\gam \nu = \pm 25$.
}
\includegraphics[width=0.475\textwidth]{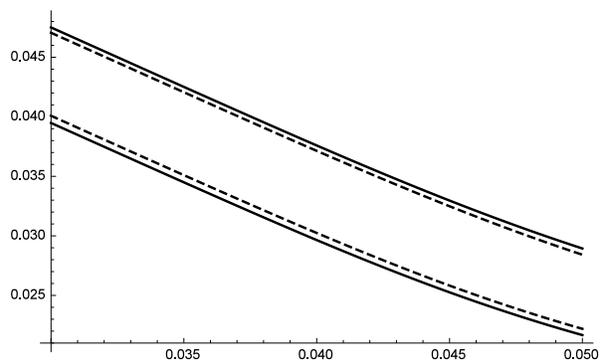}
\label{fig:VolCall2}
\end{SCfigure}

\end{document}